\pdfoutput=1
\RequirePackage{ifpdf}
\ifpdf 
\documentclass[pdftex]{sigma}
\else
\documentclass{sigma}
\fi

\numberwithin{equation}{section}

\newtheorem{Theorem}{Theorem}[section]
\newtheorem{Corollary}[Theorem]{Corollary}
\newtheorem{Lemma}[Theorem]{Lemma}
\newtheorem{Proposition}[Theorem]{Proposition}
 { \theoremstyle{definition}
\newtheorem{Definition}[Theorem]{Definition}
\newtheorem{Example}[Theorem]{Example}
\newtheorem{Remark}[Theorem]{Remark} }

\begin{document}

\allowdisplaybreaks

\newcommand{\arXivNumber}{1608.04546}

\renewcommand{\PaperNumber}{008}

\FirstPageHeading

\ShortArticleName{Classical and Quantum Superintegrability of St\"ackel Systems}

\ArticleName{Classical and Quantum Superintegrability\\ of St\"ackel Systems}

\Author{Maciej B{\L}ASZAK~$^\dag$ and Krzysztof MARCINIAK~$^\ddag$}

\AuthorNameForHeading{M.~B{\l}aszak and K.~Marciniak}

\Address{$^\dag$~Faculty of Physics, Division of Mathematical Physics, \\
\hphantom{$^\dag$}~A.~Mickiewicz University, Pozna\'{n}, Poland}
\EmailD{\href{blaszakm@amu.edu.pl}{blaszakm@amu.edu.pl}}

\Address{$^\ddag$~Department of Science and Technology, Campus Norrk\"{o}ping, Link\"{o}ping University, Sweden}
\EmailD{\href{krzma@itn.liu.se}{krzma@itn.liu.se}}

\ArticleDates{Received September 18, 2016, in f\/inal form January 19, 2017; Published online January 28, 2017}

\Abstract{In this paper we discuss maximal superintegrability of both classical and quantum St\"ackel systems. We prove a suf\/f\/icient condition for a f\/lat or constant curvature St\"ackel system to be maximally superintegrable. Further, we prove a suf\/f\/icient condition for a~St\"ackel transform to
preserve maximal superintegrability and we apply this condition to our class of St\"ackel systems, which yields new maximally superintegrable systems as conformal deformations of the original systems. Further, we demonstrate how to perform the procedure of minimal quantization to considered systems in order to produce quantum superintegrable and quantum separable systems.}


\Classification{70H06; 70H20; 81S05; 53B20}

\section{Introduction}

A real-valued function $h_{1}$ on a $2n$-dimensional manifold (phase space) $M=T^{\ast}Q$ is called a~classical maximally superintegrable Hamiltonian if it belongs to a set of $n$ Poisson-commuting functions $h_{1},\ldots,h_{n}$ (constants of motion, so that $\left\{ h_{i},h_{j}\right\} =0$ for all $i,j=1,\ldots,n$) and for which there exist $n-1$ additional functions $h_{n+1},\ldots,h_{2n-1}$ on $M$ that Poisson-commute with the Hamiltonian $h_{1}$ and such that all the functions $h_{1},\ldots,h_{2n-1}$ constitute a functionally independent set of functions. Analogously, a quantum maximally superintegrable Hamiltonian is a self-adjoint dif\/ferential operator $\hat {h}_{1}$ acting in an appropriate Hilbert space of functions on the conf\/iguration space $Q$ (square integrable with respect to some metric) belonging to a set of $n$ commuting self-adjoint dif\/ferential operators
$\hat{h}_{1},\ldots,\hat{h}_{n}$ acting in the same Hilbert space (so that $[ \hat{h}_{i},\hat{h}_{j}] =0$ for all $i,j=1,\ldots, n$) and
such that it also commutes with an additional set of $n-1$ dif\/ferential operators $\hat{h}_{n+1},\ldots,\hat{h}_{2n-1}$ of f\/inite order. Besides, in
analogy with the classical case, it is required that all the operators $\hat{h}_{1},\ldots,\hat{h}_{2n-1}$ are algebraically independent~\cite{Blaszak-5}. Throughout the paper it is tacitly assumed that $n>1$ as the case $n=1$ is not interesting from the point of view of our theory.

This paper is devoted to $n$-dimensional maximally superintegrable classical and quantum St\"ackel systems with all constants of motion quadratic in
momenta. Although superintegrable systems of second order, both classical and quantum, have been intensively studied (see for example \cite{Blaszak-bal2,Blaszak-bal3,Blaszak-3,Blaszak-1,Blaszak-2,Blaszak-4} and the review paper~\cite{Blaszak-5}), nevertheless all the results about superintegrable St\"ackel systems (including the important classif\/ication results) were mainly restricted to two or three dimensions or focused on the situation when the Hamiltonian is a~sum of one degree of freedom terms and therefore itself separates in the original coordinate system (see for example \cite{Blaszak-bal1,Blaszak-gonera} or~\cite{Blaszak-kress}). Here we present some general results concerning $n$-dimensional classical separable superintegrable systems in f\/lat spaces, constant curvature spaces and conformally f\/lat spaces. We also present how to separately quantize all considered classical systems. We stress, however, that we do not develop spectral theory of the obtained quantum systems, as it requires a~separate investigation.

The paper is organized as follows. In Section~\ref{section2} we brief\/ly describe -- following previous re\-fe\-ren\-ces, for example~\cite{Blaszak-flat} and~\cite{Blaszak-inpress}~-- f\/lat and constant curvature St\"ackel systems that we consider in this paper. In Section~~\ref{section3} we prove (Theorem~\ref{super}) a suf\/f\/icient condition for this class of St\"ackel system to be maximally superintegrable by f\/inding a~linear in momenta function $P=\sum\limits_{s=1}^{n} y^{s}p_{s}$ on $M$ such that $\{ h_{1},P\} =c$ (it also means that the vector f\/ield $Y=\sum\limits_{s=1}^{n}
y^{s}\frac{\partial}{\partial q_{s}}$ in $P$ is a Killing vector for the metric generated by $h_{1}$) which yields additional $n-1$ functions
$h_{n+i}=\{ h_{i+1},P\} $ commuting with $h_{1}$ and thus turning $h_{1}$ into a maximally superintegrable Hamiltonian. In Section~\ref{section4} we brief\/ly remind the notion of St\"ackel transform (a~functional transform that preserves integrability) and prove (Theorem~\ref{wazny}) conditions that guarantee that a St\"ackel transform transforms maximally superintegrable system into another maximally superintegrable system (i.e., preserves maximal superintegrability). In Section~\ref{section5} we apply this result to our class of maximally superintegrable St\"ackel systems, obtaining Theorem~\ref{wazny2} stating when the St\"ackel transform applied to the considered class of systems yields a St\"ackel system that is f\/lat, of constant curvature or conformally f\/lat. We also demonstrate (Theorem~\ref{tosamo}) that the additional integrals $h_{n+i}$ of systems after St\"ackel transform can be obtained in two equivalent ways. Section~\ref{section6} is devoted to the procedure of minimal quantization of considered St\"ackel systems. As the procedure of minimal quantization depends on the choice of the metric on the conf\/igurational space, we remind f\/irst the result obtained in~\cite{Blaszak-inpress} explaining how to choose the metric in which a minimal quantization is performed so that the integrability of the quantized system is preserved (Theorem~\ref{wazny3}) and then apply Lemma~\ref{lemacik2} to obtain Corollary~\ref{qint} stating under which conditions the procedure of minimal quantization of a classical St\"ackel system, considered in previous sections, yields a~quantum superintegrable and quantum separable system. The paper is furnished with several examples that continue throughout sections. The examples are all $3$-dimensional in order to make the formulas readable but our theory works in arbitrary dimension.

\section{A class of f\/lat and constant curvature St\"ackel systems}\label{section2}

Let us f\/irst introduce the class of Hamiltonian systems that we will consider in this paper. Consider a $2n$-dimensional manifold $M=T^{\ast}Q$ (we remind the reader that $n>1$) equipped with a set of (smooth) coordinates $(\lambda,\mu)=(\lambda_{1},\ldots,\lambda_{n},\mu_{1},\ldots,\mu_{n})$ def\/ined on an open dense set of $M$ and such that $\lambda$ are the coordinates on the base manifold $Q$ while $\mu$ are f\/ibre coordinates. Def\/ine the bivector
\begin{gather}
\Pi=\sum_{i=1}^{n}\frac{\partial}{\partial\lambda_{i}}\wedge\frac{\partial}{\partial\mu_{i}}.\label{bivector}
\end{gather}
Then the bivector $\Pi$ satisf\/ies the Jacobi identity so it becomes a Poisson operator (Poisson tensor), our mani\-fold~$M$ becomes Poisson manifold and the coordinates $(\lambda,\mu)$ become Darboux (canonical) coordinates for the Poisson tensor~(\ref{bivector}). Consider also a set of $n$ algebraic equations on~$M$
\begin{gather}
\sigma(\lambda_{i})+ \sum\limits_{j=1}^{n}
h_{j}\lambda_{i}^{\gamma_{j}}=\frac{1}{2}f(\lambda_{i})\mu_{i}^{2},\qquad i=1,\ldots,n, \qquad \gamma_{i}\in\mathbf{N},\label{sr}
\end{gather}
where we normalize $\gamma_{n}=0$ and where $\sigma$ and $f$ are arbitrary functions of one variable. The relations~(\ref{sr}) constitute a system of~$n$ equations linear in the unknowns~$h_{j}$. Solving these equations with respect to~$h_{j}$ we obtain $n$ functions $h_{j}=h_{j}(\lambda,\mu)$ on~$M$ of the form
\begin{gather}
h_{j}=\frac{1}{2}\mu^{T}A_{j}(\lambda)\mu+U_{j}(\lambda),\qquad j=1,\ldots,n,\label{Benham}
\end{gather}
where we denote $\lambda=(\lambda_{1},\ldots,\lambda_{n})^{T}$ and $\mu =(\mu_{1},\ldots,\mu_{n})^{T}$. The functions $h_{j}$ can be interpreted as~$n$ quadratic in momenta $\mu$ Hamiltonians on the manifold $M=T^{\ast}Q $ while the $n\times n$ symmetric matrices $A_{j}(\lambda)$ can be interpreted as~$n$ twice contravariant symmetric tensors on~$Q$. The Hamiltonians $h_{j}$ commute with respect to $\Pi$
\begin{gather*}
\{ h_{i},h_{j}\} \equiv\Pi(dh_{i},dh_{j})=0\qquad \text{for all}\quad i,j=1,\ldots,n,
\end{gather*}
since the right-hand sides of relations~(\ref{sr}) commute. Thus, the Hamiltonians in~(\ref{Benham}) constitute a Liouville integrable Hamiltonian
system (as they are moreover functionally independent). The Hamiltonians~(\ref{Benham}) constitute a wide class of the so called St\"ackel systems~\cite{Blaszak-Stackel} on $M$ while the relations~(\ref{sr}) are called separation relations~\cite{Blaszak-Sklyanin} of this system. This is the class we will consider throughout our paper. Note that by the very construction of~$h_{i}$ the variables~$(\lambda,\mu) $ are separation variables for all the Hamiltonians in~(\ref{Benham}) in the sense that the Hamilton--Jacobi equations associated with~$h_{j}$ admit a common additively separable solution.

Let us now treat the matrix $A_{1}$ as a contravariant form of a metric tensor on~$Q$: $A_{1}=G$, which turns~$Q$ into a Riemannian space. The covariant form of~$G$ will be denoted by~$g$ (so that $g=G^{-1}$). It turns out that the $(1,1)$-tensors $K_{j}$ def\/ined by
\begin{gather}
K_{j}=A_{j}g, \qquad j=1,\ldots, n\label{A}
\end{gather}
(so that $A_{j}=K_{j}G$ and $K_{1}=I$) are Killing tensors of the metric~$g$.

In this article we will focus on a particular subclass of systems~(\ref{sr}) that is given by the separation relations
\begin{gather}
\sigma(\lambda_{i})+ \sum\limits_{j=1}^{n}h_{j}\lambda_{i}^{n-j}=\frac{1}{2}f(\lambda_{i})\mu_{i}^{2},\qquad i=1,\ldots,n,\label{BenSC}
\end{gather}
(systems of the above class are known in literature as \emph{Benenti systems}) where moreover
\begin{gather}
f(\lambda)=\sum_{j=0}^{m}b_{j}\lambda^{j},\qquad b_{j}\in\mathbf{R},\quad m\in \{ 0,\ldots,n+1 \}, \label{fm}\\
\sigma(\lambda)=\sum_{k\in I}\alpha_{k}\lambda^{k},\qquad \alpha_{k} \in\mathbf{R},\label{laurent}
\end{gather}
where $I\subset\mathbf{Z}$ is some f\/inite index set (i.e., $\sigma$ is a~Laurent polynomial). Note that taking $k\in \{ 0,\ldots,n-1 \} $
will only yield trivial terms in solutions~(\ref{Benham}) of~(\ref{BenSC}), see the end of this section. Also, the parameters $\alpha_{k}$ will play a~crucial roll in the sequel, when we discuss the St\"{a}ckel transform of the above systems. The metric tensor~$G$ attains in this case, due to~(\ref{fm}), the form
\begin{gather}
G=\sum_{j=0}^{m}b_{j}G_{j}=\sum_{j=0}^{m}b_{j}L^{j}G_{0},\label{G}
\end{gather}
where $L=\operatorname{diag} ( \lambda_{1},\ldots,\lambda_{n}) $ is a~$(1,1)$-tensor (the so called special conformal Killing tensor, see for
example~\cite{Blaszak-crampin}) on~$Q$, while
\begin{gather}
G_{j}=\operatorname{diag}\left( \frac{\lambda_{1}^{j}}{\Delta_{1}},\ldots,\frac{\lambda_{n}^{j}}{\Delta_{n}}\right), \qquad j\in\mathbf{Z},\qquad \Delta_{i}= \prod\limits_{j\neq i} (\lambda_{i}-\lambda_{j}).\label{Ga}
\end{gather}

\begin{Remark}\label{remark}The metric (\ref{G}) is f\/lat for $m\leq n$ and of constant curvature for $m=n+1$ (see for example \cite[p.~788]{Blaszak-BKM}). For higher~$m$ it would have a non-constant curvature.
\end{Remark}

Further, the Killing tensors $K_{i}$ in (\ref{A}) are in this case given by
\begin{gather}
K_{i}=\sum_{r=0}^{i-1}q_{r}L^{i-1-r}=-\operatorname{diag}\left(
\frac{\partial q_{i}}{\partial\lambda_{1}},\dots,\frac{\partial q_{i}}{\partial\lambda_{n}}\right),\qquad i=1,\ldots,n.\label{Ki}
\end{gather}
Here and below $q_{i}=q_{i}(\lambda)$ are Vi\`{e}te polynomials in the variables $\lambda_{1},\ldots,\lambda_{n}$:
\begin{gather}
q_{i}(\lambda)=(-1)^{i} \sum\limits_{1\leq s_{1}<s_{2}<\cdots<s_{i}\leq n} \lambda_{s_{1}}\cdots\lambda_{s_{i}}, \qquad i=1,\ldots,n,\label{defq}
\end{gather}
that can also be considered as new coordinates on our Riemannian manifold~$Q$ (we will call them Vi\`{e}te coordinates on~$Q$). Notice that~$q_{i}$ are coef\/f\/icients of the characteristic polynomial of the tensor~$L$. Notice also that the f\/irst form of~$K_{i}$ in~(\ref{Ki}) is of course valid in any coordinate system while the second form of~$K_{i}$ is valid in separation coordinates~$\lambda$ only.

Further, due to (\ref{laurent}), the potentials $U_{j}(\lambda)$ in~(\ref{Benham}) are for the subclass~(\ref{BenSC}) given by
\begin{gather}
U_{j}=\sum_{k\in I}\alpha_{k}V_{j}^{(k)}, \qquad j=1,\ldots,n,\label{Up}
\end{gather}
where the ``basic'' potentials $V_{i}^{k}$ ($k\in\mathbf{Z}$) satisfy the linear system
\begin{gather*}
\lambda_{i}^{k}+\sum\limits_{j=1}^{n} V_{j}^{(k)}\lambda_{i}^{n-j}=0, \qquad i=1,\ldots,n, \qquad k\in\mathbf{Z},
\end{gather*}
and can be computed by the recursive formula \cite{Blaszak-reciprocal, Blaszak-macartur2011}
\begin{gather}
V^{(k)}=F^{k}V^{(0)},\qquad k\in\mathbf{Z},\label{U}
\end{gather}
where $V^{(k)}=\big( V_{1}^{(k)},\ldots,V_{n}^{(k)}\big) ^{T}$, $V^{(0)}=(0,0,\ldots,0,-1)^{T}$ and where $F$ is an $n\times n$ matrix given by
\begin{gather}
F=\left(
\begin{matrix}
-q_{1}(\lambda) & 1 & & \\
-q_{2}(\lambda) & & \ddots & \\
\vdots & & & 1\\
-q_{n}(\lambda) & 0 & \cdots & 0
\end{matrix}
\right) \label{R}
\end{gather}
with $q_{i}(\lambda)$ given by~(\ref{defq}). Note that the formulas (\ref{U}),~(\ref{R}) are non tensorial in that they are the same in an arbitrary coordinate system, not only in the separation variables~$\lambda_{i}$. As we mentioned above, the f\/irst potentials, i.e.,~$V^{(1)}=(0,0,\ldots,0,-1,0)^{T}$ up to $V^{(n-1)}=(-1,0,\ldots,0)^{T}$ are constant, $V^{(n)}=(q_{1},\ldots,q_{n})$ is the f\/irst nonconstant positive potential while $V^{(-1)}$ $=(1/q_{n},q_{1}/q_{n},\ldots,q_{n-1}/q_{n})^{T}$. The potentials~$V^{(k)}$ are for $k<0$ rational functions of~$q$ that quickly become complicated with decreasing~$k$.

To summarize, the Hamiltonians $h_{i}$ generated by (\ref{BenSC})--(\ref{laurent}) can be explicitly written as
\begin{gather*}
h_{r}(\lambda)=-\frac{1}{2}\sum_{i=0}^{n}\frac{\partial q_{r}}{\partial \lambda_{i}}\frac{f(\lambda_{i})\mu_{i}^{2}-\sigma(\lambda_{i})}{\Delta_{i}
}=-\frac{1}{2}\sum_{i=0}^{n}\frac{\partial q_{r}}{\partial\lambda_{i}} \frac{f(\lambda_{i})\mu_{i}^{2}}{\Delta_{i}}+U_{r}(\lambda),\qquad
r=1,\ldots,n.
\end{gather*}

\section[Maximally superintegrable f\/lat and constant curvature St\"ackel systems]{Maximally superintegrable f\/lat and constant curvature\\ St\"ackel systems}\label{section3}

Suppose that we have an integrable system, i.e., $n$ functionally independent Hamiltonians on a~$2n$-dimensional phase space $M$ that pairwise commute: $\{ h_{i},h_{j}\} =0$ for all $i,j=1,\ldots,n$. If there exists an additional function $P$ commuting to a constant with one of the Hamiltonians, say with $h_{1}$ (so that $\{ h_{1},P\} =c$) and if the $n-1$ functions
\begin{gather*}
h_{n+i}= \{ h_{i+1},P \} ,\qquad i=1,\ldots,n-1
\end{gather*}
together with all $h_{i}$ are functionally independent, then the system becomes maximally superintegrable (with respect to this particular Hamiltonian~$h_{1}$) since then by the Jacobi identity
\begin{gather*}
 \{ h_{n+i},h_{1} \} =- \{ \{ P,h_{1} \},h_{i+1} \} - \{ \{ h_{1},h_{i+1} \} ,P \}=0,\qquad i=1,\ldots,n-1.
\end{gather*}
If moreover the f\/irst $n$ integrals of motion~$h_{i}$ are quadratic in momenta and if $P$ is linear in momenta, then the resulting $n-1$ extra integrals of motion~$h_{n+i}$ are also quadratic in momenta. Thus, in order to distinguish those constant curvature St\"ackel systems that are maximally superintegrable and have quadratic in momenta extra integrals of motion we have to f\/ind~$P$ that commutes with~$h_{1}$ up to a constant and that is linear in momenta. To do it in a systematic way, we need the following well-known result.

\begin{Lemma}\label{lemacik}Suppose that $(q,p)=(q_{1},\ldots,q_{n},p_{1},\ldots,p_{n})$ are Darboux $($canonical$)$ coordinates on a~$2n$-dimensional phase space $M=T^{\ast}Q$. Consider two functions on~$M$:
\begin{gather*}
h=\frac{1}{2}\sum_{i,j=1}^{n}p_{i}A^{ij}(q)p_{j}+U(q)\qquad \text{with} \quad A=A^{T} \qquad \text{and} \qquad P=\sum\limits_{i=1}^{n}y^{i}(q)p_{i}.
\end{gather*}
Then%
\begin{gather*}
\{ h,P\} =\frac{1}{2}\sum_{i,j=1}^{n}p_{i}( L_{Y}A)^{ij}p_{j}+Y(U),
\end{gather*}
where $Y$ is the vector field on $Q$ given by $Y=\sum\limits_{i=1}^{n}y^{i}(q)\frac{\partial}{\partial q_{i}}$ and where $L_{Y}$ is the Lie
derivative $($on~$Q)$ along~$Y$.
\end{Lemma}

One can thus say that $h$ and $P$ commute if the corresponding vector f\/ield~$Y$ is the Killing vector for the metric def\/ined by the $(2,0)$-tensor~$A$ (i.e., if $L_{Y}A=0$) and if moreover~$Y$ is symmetry of~$U$ (i.e., if $Y(U)=0$).

Consider now the St\"ackel system given by~(\ref{BenSC}). The coordinates $(\lambda,\mu)$ are Darboux (so that the above lemma applies to this
situation) but the components of the metric~(\ref{G}) expressed in $\lambda$-coordinates are rational functions making computations very complicated. We will therefore perform the search for the function~$P$ in the coordinates $(q,p)$ on $M$ such that $q_{i}$ are Vi\`{e}te coordinates~(\ref{defq}) and such that
\begin{gather}
p_{i}=-\sum_{k=1}^{n}\frac{( \lambda_{k}) ^{n-i}\mu_{k}}{\Delta_{k}}\label{defp}
\end{gather}
are the conjugated momenta. Since the transformation from $(\lambda,\mu)$ to $(q,p)$ is a point transformation the coordinates $(q,p)$ are also Darboux coordinates four our Poisson tensor. It can be shown~\cite{Blaszak-BS} that in the $(q,p)$-coordinates
\begin{gather}
(L)_{j}^{i}=-\delta_{j}^{1}q_{i}+\delta_{j}^{i+1},\qquad (G_{0}) ^{ij}=\sum_{k=0}^{n-1}q_{k}\delta_{n+k+1}^{i+j},\label{LG}
\end{gather}
and moreover
\begin{gather}
( G_{r}) ^{ij} =
\begin{cases}
\sum\limits_{k=0}^{n-r-1}q_{k}\delta_{n-r+k+1}^{i+j},& i,j=1,\dots,n-r,\\
-\sum\limits_{k=n-r+1}^{n}q_{k}\delta_{n-r+k+1}^{i+j},& i,j=n-r+1,\dots,n,\\
0 & \text{otherwise},
\end{cases} \qquad r=1,\dots,n,\label{Gr}\\
( G_{n+1}) ^{ij} =q_{i}q_{j}-q_{i+j},\qquad i,j=1,\dots,n,\nonumber
\end{gather}
where we set $q_{0}\equiv1$ and $q_{r}=0$ for $r>n$. An advantage of these new coordinates is that the geodesic parts of~$h_{i}$ are polynomial in~$q$.

\begin{Example}
For $n=3$ and in Vi\`{e}te coordinates (\ref{defq}) we have
\begin{gather}
L=\left(
\begin{matrix}
-q_{1} & 1 & 0\\
-q_{2} & 0 & 1\\
-q_{3} & 0 & 0
\end{matrix}
\right) ,\qquad G_{0}=\left(
\begin{matrix}
0 & 0 & 1\\
0 & 1 & q_{1}\\
1 & q_{1} & q_{2}
\end{matrix}
\right), \label{G0}
\end{gather}
and hence the metric tensors $G_{j}$ have the form
\begin{gather}
G_{1}=\left(
\begin{matrix}
0 & 1 & 0\\
1 & q_{1} & 0\\
0 & 0 & -q_{3}
\end{matrix}
\right) ,\qquad G_{2}=\left(
\begin{matrix}
1 & 0 & 0\\
0 & -q_{2} & -q_{3}\\
0 & -q_{3} & 0
\end{matrix}
\right) ,\label{G12}
\\
G_{3}=\left(
\begin{matrix}
-q_{1} & -q_{2} & -q_{3}\\
-q_{2} & -q_{3} & 0\\
-q_{3} & 0 & 0
\end{matrix}
\right) ,\qquad G_{4}=\left(
\begin{matrix}
q_{1}^{2}-q_{2} & q_{1}q_{2}-q_{3} & q_{1}q_{3}\\
q_{1}q_{2}-q_{3} & q_{2}^{2} & q_{2}q_{3}\\
q_{1}q_{3} & q_{2}q_{3} & q_{3}^{2}
\end{matrix}
\right). \label{G34}
\end{gather}
In accordance with Remark \ref{remark}, the metric tensors $G_{0},\dots,G_{3}$ are f\/lat, while the metric $G_{4}$ is of constant curvature.
\end{Example}

We are now in position to perform our search for $P$. We do this in the case when $\sigma$ is the Laurent polynomial (\ref{laurent}) and allow $f$ to be polynomial as in~(\ref{fm}); a particular case of $f=\lambda^{m}$ of the theorem below was formulated in~\cite{Blaszak-supBenMA}.

\begin{Theorem}\label{super} The St\"ackel system
\begin{gather*}
\sum_{k\in I}\alpha_{k}\lambda_{i}^{k}+ \sum\limits_{j=1}^{n}
h_{j}\lambda_{i}^{n-j}=\frac{1}{2}f(\lambda_{i})\mu_{i}^{2},\qquad i=1,\ldots,n,
\end{gather*}
$($where $I\subset\mathbf{Z}$ is a finite index set$)$ with $f(\lambda_{i})$ given by
\begin{gather}
f(\lambda)=\sum_{j=0}^{m}b_{j}\lambda^{j}, \qquad b_{j}\in\mathbf{R},\qquad m\in \{ 0,\ldots,n+1 \} \label{bj}
\end{gather}
is maximally superintegrable in the following cases:
\begin{enumerate}\itemsep=0pt
\item[$(i)$] case $m\in \{ 0,\ldots,n-1 \}$: if $I\subset \{n,\dots,2n-m-1 \} \cup \{ -1,\dots,-r-1 \} $, where~$r$ is such
that $b_{i}=0$ for $i=0,\dots,r\leq m-1$ $($if all \thinspace$b_{i}\neq0$, then there is no such $r$ and no second component in~$I)$;

\item[$(ii)$] case $m=n$ and $b_{0}=b_{1}=0$: if $I\subset \{n,-1,\dots,-r+1 \} $, where $r$ is such that $b_{i}=0$ for $i=2,\dots,r\leq
n-1$;

\item[$(iii)$] case $m=n+1$ $($case of constant curvature$)$ and $b_{0}=b_{1}=0$: if $I\subset \{ -1,\dots,-r+1 \} $, where~$r$ is such that $b_{i}=0$ for $i=2,\dots,r\leq n$.
\end{enumerate}

The additional integrals $h_{n+r}$ commuting with $h_{1}$ are given in $(q,p)$-coordinates by
\begin{gather}
h_{n+r}=\frac{1}{2}\sum_{i,j=1}^{n}p_{i} ( L_{Y}A_{r+1} )^{ij}p_{j}+Y(U_{r+1}) , \qquad r=1,\ldots,n-1,\label{struk}
\end{gather}
where $Y$ is a vector field on $Q$ given by
\begin{enumerate}\itemsep=0pt
\item[$(i)$] for $m\in \{ 0,\ldots,n-1 \}$
\begin{gather}
Y=\sum\limits_{i=0}^{m}b_{m-i}\frac{\partial}{\partial q_{n-m+i}},\label{jeden}
\end{gather}

\item[$(ii)$] for $m=n$
\begin{gather}
Y=q_{n}\sum\limits_{i=2}^{n}b_{n-i+2}\frac{\partial}{\partial q_{i}},\label{dwa}
\end{gather}

\item[$(iii)$] for $m=n+1$%
\begin{gather}
Y=q_{n}\sum\limits_{i=1}^{n}b_{n-i+2}\frac{\partial}{\partial q_{i}},\label{trzy}
\end{gather}
\end{enumerate}
and where $L_{Y}$ denotes the Lie derivative along $Y$.
\end{Theorem}

\begin{proof}We will search for a function $P$ that commutes with $h_{1}$ and we will perform this search in the $(q,p)$-coordinates (\ref{defq}), (\ref{defp}). The Hamiltonian $h_{1}$ has in these coordinates the form
\begin{gather*}
h_{1}=\frac{1}{2}\sum_{i,j=1}^{n}G^{ij}(q)p_{i}p_{j}+V_{1}^{(k)}(q)
\end{gather*}
with $G$ given by (\ref{G}) and further by (\ref{LG}), (\ref{Gr}) and with the potential $V_{1}^{(k)}(q)$ def\/ined by~(\ref{U}) and~(\ref{R}).

$(i)$ For $m=0,\dots,n-1$, the Killing equation $L_{Y}G=0$ has a unique (up to a multiplicative constant) constant solution~(\ref{jeden}) which also satisf\/ies $Y\big( V_{1}^{(k)}\big) =0$ for $k=n,\dots,2n-m-2$ and $Y\big( V_{1}^{(2n-m-1)}\big) =c$. In consequence, due to Lemma~\ref{lemacik}, the function
\begin{gather*}
P=b_{m}p_{n-m}+b_{m-1}p_{n-m+1}+\dots +b_{0}p_{n}
\end{gather*}
satisf\/ies
\begin{gather}
\{h_{1},P\}=0 \label{jola}
\end{gather}
for $k=n,\dots,2n-m-2$ and
\begin{gather*}
\{h_{1},P\}=c
\end{gather*}
for $k=2n-m-1$. Moreover, if $b_{i}=0$ for $i=0,\dots,r\leq m-1,$ then~(\ref{jola}) is satisf\/ied also for $k=-1,\dots,-r-1$.

$(ii)$ For $m=n$, there is no constant solution of $L_{Y}G=0$. This equation has a~simple linear in $q$ solution~(\ref{dwa}) provided that $b_{0}=b_{1}=0$; $Y$ is then also a symmetry for the single nontrivial potential $V_{1}^{(n)}$, i.e., $Y\big( V_{1}^{(n)}\big) =0$. In consequence, the function
\begin{gather*}
P=q_{n} ( b_{n}p_{2}+b_{n-1}p_{3}+\dots +b_{2}p_{n} )
\end{gather*}
satisf\/ies (\ref{jola}). Moreover, if $b_{i}=0$ for $i=2,\dots,r\leq n-1$ then~(\ref{jola}) is satisf\/ied also for $k=-1,\dots,-r+1$.

$(iii)$ For $m=n+1$ there is no constant solution of $L_{Y}G=0$. This equation has a~simple linear in $q$ solution~(\ref{trzy}) provided that $b_{0}=b_{1}=0$ but $Y$ is not a symmetry for any nontrivial potential $V_{1}^{(k)}$. In consequence, the function
\begin{gather*}
P=q_{n} ( b_{n+1}p_{1}+b_{n}p_{2}+\dots +b_{2}p_{n} ).
\end{gather*}
Poisson commutes only with the geodesic part $E_{1}$ of $h_{1}$: $\{E_{1},P\}=0$. However, if $b_{i}=0$ for $i=2,\dots,r\leq n$ then~(\ref{jola})
is satisf\/ied for $k=-1,\dots,-r+1$.

\looseness=1 Finally, the form of additional integrals $h_{n+r}$ in (\ref{struk}) is obtained through $h_{n+r}= \{ h_{r+1},P \} $ by~using Lemma~\ref{lemacik}. Due to their form, the functions $h_{1},\ldots,h_{2n-1}$ are functionally inde\-pen\-dent.
\end{proof}

\begin{Remark}\label{parametry}The above theorem provides us with a suf\/f\/icient condition for maximal superintegrability of St\"ackel systems of constant curvature (f\/lat in particular) in case when~$f(\lambda)$ is a~polynomial of maximal order~$n+1$. In consequence, the case~$(i)$ of Theorem~\ref{super} yields an $( n+1) $-parameter family of maximally superintegrable systems, parametrized by
\begin{gather*}
 \{ b_{r},\dots,b_{m},\alpha_{-r-1},\dots,\alpha_{-1},\alpha_{n},\dots,\alpha_{2n-m-1} \} ,\qquad r=0,\dots,m,
\end{gather*}
where $b_{j}$ parametrize superintegrable metrics (\ref{fm}), (\ref{G}) and~$\alpha_{j}$ parametrize families of nontrivial superintegrable potentials~$U$~(\ref{Up}) (in case there is no~$r$, i.e., all $b_{i}\neq0$) then there is no~$\alpha_{-j}$ in the above set. Similarly, in the cases $(ii)$ and $(iii)$ Theorem~\ref{super} yields appropriate $n$-parameter families of superintegrable systems. A particular case of that classif\/ication (for the monomial case $f(\lambda)=\lambda^{m}$) was presented in~\cite{Blaszak-Artur}.
\end{Remark}

It is possible to calculate explicitly the structure of the geodesic parts $E_{n+r}$ of the extra integrals~$h_{n+r}$ in the separation coordinates~$(\lambda,\mu)$.

\begin{Proposition}The geodesic parts
\begin{gather*}
E_{n+r}=\frac{1}{2}\sum_{i,j=1}^{n}\mu_{i}A_{n+r}^{ij}(\lambda)\mu_{j},\qquad r=1,\ldots,n-1
\end{gather*}
of additional integrals of motion $h_{n+r}= \{ h_{r+1},P \} $ with $r=1,\ldots,n-1$ are given by
\begin{enumerate}\itemsep=0pt
\item[$(i)$] for $0\leq m\leq n-1$
\begin{gather*}
A_{n+r}^{ij} =-\frac{\partial^{2}q_{r}}{\partial\lambda_{i}\partial \lambda_{j}}\frac{f(\lambda_{i})f(\lambda_{j})}{\Delta_{i}\Delta_{j}},\qquad i\neq j,\\
A_{n+r}^{ii} =\frac{f(\lambda_{i})}{\Delta_{i}}\sum\limits_{j=1}^{n}\frac{\partial^{2}q_{r}}{\partial\lambda_{i}\partial\lambda_{j}}
\frac{f(\lambda_{j})}{\Delta_{j}},
\end{gather*}
where $q_{r}=q_{r}(\lambda)$ are given by \eqref{defq}, $f(\lambda)$ are given by \eqref{fm} while $\Delta_{i}$ by~\eqref{Ga},

\item[$(ii)$] for $m=n,n+1$
\begin{gather*}
A_{n+r}^{ij} =-\frac{\partial^{2}q_{r}}{\partial\lambda_{i}\partial \lambda_{j}}\frac{\partial q_{n}}{\partial\lambda_{i}}\frac{1}{\lambda_{j}
}\frac{f(\lambda_{i})f(\lambda_{j})}{\Delta_{i}\Delta_{j}}, \qquad i\neq j,\\
A_{n+r}^{ii} =\frac{f(\lambda_{i})}{\Delta_{i}}\sum\limits_{j=1}^{n}\frac{\partial^{2}q_{r}}{\partial\lambda_{i}\partial\lambda_{j}}
\frac{f(\lambda_{j})}{\Delta_{j}}\frac{\partial q_{n}}{\partial\lambda_{j}}\frac{1}{\lambda_{j}}.
\end{gather*}
\end{enumerate}
\end{Proposition}

Let us illustrate the above considerations by some examples.

\begin{Example}\label{1}Consider the f\/lat case $n=3$, $m=1$, $b_{1}=1$, (so that $f(\lambda)=b_{0}+\lambda$) with $\sigma(\lambda)= \alpha\lambda^{k}$ and where $k=-1,3$ or~$4$. The commuting Hamiltonians~$h_{i}$ are given by separation relations~(\ref{BenSC})
\begin{gather*}
\alpha\lambda_{i}^{k}+h_{1}\lambda_{i}^{2}+h_{2}\lambda_{i}+h_{3}=\frac{1}{2}(b_{0}+\lambda_{i})\mu_{i}^{2}, \qquad i=1,2,3.
\end{gather*}
Then, according to (\ref{Ki}), (\ref{G0}), (\ref{G34}) and to (\ref{U}), (\ref{R}) the corresponding St\"ackel Hamiltonians attain in the $(q,p)$
coordinates (\ref{defq}), (\ref{defp}) the form
\begin{gather*}
h_{1} =p_{1}p_{2}+b_{0}p_{1}p_{3}+b_{0}q_{1}p_{2}p_{3}+\frac{1}{2}(q_{1}+b_{0})p_{2}^{2}+\frac{1}{2}(b_{0}q_{2}-q_{3})p_{3}^{2}+\alpha V_{1}^{(k)}(q),\\
h_{2} =\frac{1}{2}p_{1}^{2}+\frac{1}{2}\big(q_{1}^{2}+2b_{0}q_{1}-q_{2}\big)p_{2}^{2}+\frac{1}{2}(b_{0}q_{1}q_{2}-q_{1}q_{3}-b_{0}q_{3})p_{3}^{2} +(q_{1}+b_{0})p_{1}p_{2}\\
\hphantom{h_{2} =}{} +b_{0}q_{1}p_{1}p_{3}+\big(b_{0}q_{1}^{2}-q_{3}\big)p_{2}p_{3}+\alpha V_{2}^{(k)}(q),\\
h_{3} =\frac{1}{2}b_{0}p_{1}^{2}+\frac{1}{2}(b_{0}q_{1}^{2}-q_{3})p_{2} ^{2}+\frac{1}{2}\big({-}b_{0}q_{1}q_{3}+b_{0}q_{2}^{2}-q_{2}q_{3}\big)p_{3}^{2} +b_{0}q_{1}p_{1}p_{2}\\
\hphantom{h_{3} =}{} +(b_{0}q_{2}-q_{3})p_{1}p_{3}+(b_{0}q_{1}q_{2}-q_{1}q_{3}-b_{0}q_{3} )p_{2}p_{3}+\alpha V_{3}^{(k)}(q),
\end{gather*}
where
\begin{gather*}
V_{1}^{(-1)}=\frac{1}{q_{3}},\qquad V_{2}^{(-1)}=\frac{q_{1}}{q_{3}},\qquad V_{3}^{(-1)}=\frac{q_{2}}{q_{3}},
\\
V_{1}^{(3)}=q_{1},\qquad V_{2}^{(3)}=q_{2},\qquad V_{3}^{(3)}=q_{3},
\\
V_{1}^{(4)}=-q_{1}^{2}+q_{2},\qquad V_{2}^{(4)}=-q_{1}q_{2}+q_{3},\qquad V_{3}^{(4)}=-q_{1}q_{3}.
\end{gather*}
According to Theorem \ref{super} $Y=\frac{\partial}{\partial q_{2}}+b_{0}\frac{\partial}{\partial q_{3}}$ so that $P=p_{2}+b_{0}p_{3}$ and thus
\begin{gather*}
\{h_{1},P\}=
\begin{cases} 0 & \text{for} \ k=-1 \ \text{and} \ b_{0}=0,\\
0 & \text{for} \ k=3,\\
\alpha & \text{for} \ k=4 \ \big(\text{then} \ Y\big(V_{1}^{(4)}\big)=1\big).
\end{cases}
\end{gather*}
Hence, \looseness=-1 the system is maximally superintegrable with additional constants of motion for $h_{1}$ given~by:

\noindent for $k=-1$ and $b_{0}=0$
\begin{gather*}
h_{4} =\{ h_{2},P\} =-\frac{1}{2}p_{2}^{2}, \qquad h_{5} =\{h_{3},P\} =-\frac{1}{2}q_{3}p_{3}^{2}+\alpha \frac{q_{3}}{q_{3}^{2}},
\end{gather*}
for $k=3$
\begin{gather*}
\begin{split}
& h_{4} =\{ h_{2},P\} =-\frac{1}{2}p_{2}^{2}-\frac{1}{2}b_{0}^{2}p_{3}^{2}-b_{0}p_{2}p_{3}+\alpha,\\
& h_{5} = \{h_{3},P\} =-\frac{1}{2}b_{0}p_{2}^{2}+\left(\frac{1}{2}b_{0}q_{2}-\frac{1}{2}q_{3}-\frac{1}{2}b_{0}^{2}q_{1}\right)p_{3}^{2} -b_{0}^{2}p_{2}p_{3}+\alpha b_{0},
\end{split}
\end{gather*}
and for $k=4$
\begin{gather*}
h_{4} =\{ h_{2},P\} =-\frac{1}{2}p_{2}^{2}-\frac{1}{2}b_{0}^{2}p_{3}^{2}-b_{0}p_{2}p_{3}+\alpha(b_{0}-q_{1}),\\
h_{5} =\{ h_{3},P\} =-\frac{1}{2}b_{0}p_{2}^{2}+\frac{1}{2}(b_{0}q_{2}-q_{3}-b_{0}^{2}q_{1})p_{3}^{2}-b_{0}^{2}p_{2}p_{3} -\alpha b_{0}q_{1}.
\end{gather*}
\end{Example}

\begin{Example} \label{2}Consider the case $n=3,$ $m=1$, with the monomial $f(\lambda)=\lambda$, given by the separation relations
\begin{gather}
\alpha_{4}\lambda_{i}^{4}+\alpha_{3}\lambda_{i}^{3}+h_{1}\lambda_{i}^{2}%
+h_{2}\lambda_{i}+h_{3}+\alpha_{-1}\lambda_{i}^{-1}=\frac{1}{2}\lambda_{i}%
\mu_{i}^{2} , \qquad i=1,2,3,\label{sep}
\end{gather}
so that $I= \{ -1,3,4 \} $ and satisf\/ies the condition in part~$(i)$ of Theorem~\ref{super}. The system is thus maximally superintegrable and
has a three-parameter family of potentials (cf.\ Remark~\ref{parametry}). Consider now the point transformation from $(q,p)$-coordinates (\ref{defq}), (\ref{defp}) to non-orthogonal coordinates $(r,s)$ such that $r_{i}$ are given by~\cite{Blaszak-BS}
\begin{gather}\label{r}
q_{1} =r_{1},\qquad q_{2} =r_{2}+\frac{1}{4}r_{1}^{2}, \qquad q_{3} =-\frac{1}{4}r_{3}^{2},
\end{gather}
while
\begin{gather}
s_{j}=\sum_{i=1}^{3}\frac{\partial q_{i}}{\partial r_{j}}p_{i},\qquad j=1,2,3\label{s}
\end{gather}
are new conjugated momenta. Then $r_{i}$ are f\/lat coordinates for the metric $G_{1}=A_{1}$ in~$h_{1}$. In these coordinates we get in this case
\begin{gather}
G=G_{1}=\left(
\begin{matrix}
0 & 1 & 0\\
1 & 0 & 0\\
0 & 0 & 1
\end{matrix}
\right) ,\qquad L=\left(
\begin{matrix}
-\frac{1}{2}r_{1} & 1 & 0\\
-r_{2} & -\frac{1}{2}r_{1} & -\frac{1}{2}r_{3}\\
-\frac{1}{2}r_{3} & 0 & 0
\end{matrix}
\right), \label{Go}
\end{gather}
while the f\/irst three commuting Hamiltonians in $(r,s)$-variables become
\begin{gather}
h_{1} =s_{1}s_{2}+\frac{1}{2}s_{3}^{2}+\alpha_{-1}V_{1}^{(-1)}(r)+\alpha_{3}V_{1}^{(3)}(r)+\alpha_{4}V_{1}^{(4)}(r),\nonumber\\
h_{2} =\frac{1}{2}s_{1}^{2}-\frac{1}{2}r_{2}s_{2}^{2}+\frac{1}{2}r_{1}
s_{3}^{2}+\frac{1}{2}r_{1}s_{1}s_{2}-\frac{1}{2}r_{3}s_{2}s_{3}+\alpha
_{-1}V_{2}^{(-1)}(r)+\alpha_{3}V_{2}^{(3)}(r)+\alpha_{4}V_{2}^{(4)}(r),\nonumber \\
h_{3} =\frac{1}{8}r_{3}^{2}s_{2}^{2}+\left(\frac{1}{2}r_{2}+\frac{1}{8}r_{1}
^{2}\right)s_{3}^{2}-\frac{1}{2}r_{3}s_{1}s_{3}-\frac{1}{4}r_{1}r_{3}s_{2}s_{3}\nonumber\\
\hphantom{h_{3} =}{} +\alpha_{-1}V_{3}^{(-1)}(r)+\alpha_{3}V_{3}^{(3)}(r)+\alpha_{4}
V_{3}^{(4)}(r), \label{ham2}
\end{gather}
with
\begin{gather}
V_{1}^{(-1)}=\frac{4}{r_{3}^{2}},\qquad V_{2}^{(-1)}=\frac{4r_{1}}{r_{3}^{2}},\qquad V_{3}^{(-1)}=\frac{r_{1}^{2}+4r_{2}}{r_{3}^{2}},\label{pot1}
\\
V_{1}^{(3)}=r_{1},\qquad V_{2}^{(3)}=\left( r_{2}+\frac{1}{4}r_{1}^{2}\right),\qquad V_{3}^{(3)}=-\frac{1}{4}r_{3}^{2},\label{pot2}
\\
V_{1}^{(4)}=r_{2}-\frac{3}{4}r_{1}^{2},\qquad V_{2}^{(4)}=-\left( r_{1}r_{2}+\frac{1}{4}r_{1}^{3}+\frac{1}{4}r_{3}^{2}\right) ,\qquad V_{3}^{(4)}=\frac{1}{4}r_{1}r_{3}^{2}.\label{pot3}
\end{gather}
In accordance with Theorem~\ref{super} and after the transformation to $(r,s)$-coordinates we have $P=s_{2}$, and $Y=\frac{\partial}{\partial r_{2}}$ so the additional constants of motion $h_{n+i}$ of $h_{1}$ are
\begin{gather}
h_{4}= \{ h_{2},P \} =-\frac{1}{2}s_{2}^{2}+\alpha_{3}-\alpha _{4}r_{1},\qquad h_{5}= \{ h_{3},P \} =\frac{1}{2}s_{3}^{2}+\frac{4\alpha_{-1}}{r_{3}^{2}}. \label{jola1}
\end{gather}
\end{Example}

\begin{Example}\label{3}Consider the constant curvature case $n=3,$ $m=4$ and $I=\{-2,-1\} $. In order to apply part~$(iii)$ of Theorem~\ref{super} we have to put $b_{0}=b_{1}=0$. Assume further that also $b_{2}=b_{3}=0$ and $b_{4}=1$ (so that $f(\lambda)=\lambda^{4}$ is again a monomial). The commuting Hamiltonians are then given by the separation relations
\begin{gather*}
\alpha_{-2}\lambda_{i}^{-2}+\alpha_{-1}\lambda_{i}^{-1}+h_{1}\lambda_{i}^{2}+h_{2}\lambda_{i}+h_{3}=\frac{1}{2}\lambda_{i}^{4}\mu_{i}^{2},
\qquad i=1,2,3.
\end{gather*}
Then again, according to (\ref{Ki}), (\ref{G0})--(\ref{G34}) and to (\ref{U}), (\ref{R}), the corresponding St\"ackel Hamiltonians attain in the
$(q,p)$-variables the form
\begin{gather*}
h_{1} =\frac{1}{2}\big(q_{1}^{2}-q_{2}\big)p_{1}^{2}+\frac{1}{2}q_{2}^{2}p_{2}^{2}+\frac{1}{2}q_{3}^{2}p_{3}^{2}+(q_{1}q_{2} -q_{3})p_{1}p_{2}+q_{1}q_{3}p_{1}p_{3}+q_{2}q_{3}p_{2}p_{3}\\
\hphantom{h_{1} =}{} +\alpha_{-2}V_{1}^{(-2)}+\alpha_{-1}V_{1}^{(-1)},\\
h_{2} =\frac{1}{2}(q_{1}q_{2}-q_{3})p_{1}^{2}+q_{2}q_{3}p_{2}^{2}+q_{2}^{2}p_{1}p_{2} +q_{2}q_{3}p_{1}p_{3}+q_{3}^{2}p_{2}p_{3}+\alpha_{-2}V_{2}^{(-2)}+\alpha_{-1}V_{2}^{(-1)},\\
h_{3} =\frac{1}{2}q_{1}q_{3}p_{1}^{2}+\frac{1}{2}q_{3}^{2}p_{2}^{2}+q_{2}q_{3}p_{1}p_{2}+q_{3}^{2}p_{1}p_{3}+\alpha_{-2}V_{3}^{(-2)} +\alpha_{-1}V_{3}^{(-1)}
\end{gather*}
with
\begin{gather*}
V_{1}^{(-1)} =\frac{1}{q_{3}},\qquad V_{2}^{(-1)}=\frac{q_{1}}{q_{3}},\qquad V_{3}^{(-1)}=\frac{q_{2}}{q_{3}},\\
V_{1}^{(-2)} =-\frac{q_{2}}{q_{3}^{2}},\qquad V_{2}^{(-2)}=\frac{1}{q_{3}}-\frac{q_{1}q_{2}}{q_{3}^{2}},\qquad V_{3}^{(-2)}=\frac{q_{1}}{q_{3}}-\frac{q_{2}^{2}}{q_{3}^{2}}.
\end{gather*}
Now, according to part $(iii)$ of Theorem~\ref{super}, $P=q_{3}p_{1}$, $Y=q_{3}\frac{\partial}{\partial q_{1}}$ and $\{h_{1},P\}=0$ so the additional constants of motion are
\begin{gather*}
h_{4}=\{h_{2},P\}=-\frac{1}{2}q_{2}q_{3}p_{1}^{2}-q_{3}^{2}p_{1}p_{2} +\alpha_{-1}-\alpha_{-2}\frac{q_{2}}{q_{3}}, \qquad h_{5}=\{h_{3},P\}=-\frac{1}{2}q_{3}^{2}p_{1}^{2}+\alpha_{-2}.
\end{gather*}
\end{Example}

\section{St\"ackel transforms preserving maximal superintegrability}\label{section4}

In this chapter we apply a $1$-parameter St\"ackel transform to our systems (\ref{BenSC})--(\ref{laurent}) to produce new maximally superintegrable St\"ackel systems. As the transformation parameter~$\alpha$ we will always use one of the~$\alpha_{i}$ from~(\ref{laurent}).

St\"ackel transform is a functional transform that maps a Liouville integrable systems into a new integrable system. It was f\/irst introduced in~\cite{Blaszak-hietarinta} (where it was called the coupling-constant metamorphosis) and later developed in~\cite{Blaszak-BKM}. When applied to a~St\"ackel separable system, this transformation yields a new St\"ackel separable system, which explains its name. In the original paper~\cite{Blaszak-hietarinta} the authors used only one parameter (one coupling constant). In~\cite{Blaszak-macartur2008} the authors introduced a~multiparameter generalization of this transform. This idea has been further developed in~\cite{Blaszak-macartur2011} and later in~\cite{Blaszak-reciprocal}.

In this section we prove a theorem (Theorem~\ref{wazny}) that yields suf\/f\/icient conditions for St\"ackel transform to preserve maximal superintegrability of a St\"ackel system.

Let us f\/irst, following \cite{Blaszak-reciprocal}, remind the def\/inition of the multiparameter St\"ackel transform. Consider again a manifold $M$ equipped with a Poisson tensor $\Pi$ and the corresponding Poisson bracket $\{\cdot,\cdot\} $. Suppose we have $r$ smooth functions $h_{i}\colon M\rightarrow\mathbf{R}$ on $M,$ each depending on $k\leq r$ parameters $\alpha_{1},\dots,\alpha_{k}$ so that
\begin{gather}
h_{i}=h_{i}(x,\alpha_{1},\dots,\alpha_{k}),\qquad i=1,\dots,r,\label{hi}
\end{gather}
where $x\in M$. Let us now from $r$ functions in (\ref{hi}) choose~$k$ functions $h_{s_{i}}$, $i=1,\ldots,k$, where $\{s_{1},\dots,s_{k}\}$
$\subset\{1,\dots,r\}$. Assume also that the system of equations
\begin{gather*}
h_{s_{i}}(x,\alpha_{1},\dots,\alpha_{k})=\tilde{\alpha}_{i},\qquad i=1,\dots,k,
\end{gather*}
(where $\tilde{\alpha}_{i}$ is another set of $k$ free parameters, or values of Hamiltonians $h_{s_{i}}$) involving the functions $h_{s_{i}}$ can be solved for the parameters $\alpha_{i}$ yielding
\begin{gather}
\alpha_{i}=\tilde{h}_{s_{i}}(x,\tilde{\alpha}_{1},\dots,\tilde{\alpha}_{k}),\qquad i=1,\dots,k,\label{alphy}
\end{gather}
where the right hand sides of these solutions def\/ine $k$ new functions $\tilde{h}_{s_{i}}$ on $M$, each depending on $k$ parameters $\tilde{\alpha
}_{i}$. Finally, let us def\/ine $r-k$ functions $\tilde{h}_{i}$ for $i=1,\ldots,r$, $i\notin\{s_{1},\dots,s_{k}\}$, by substituting $\tilde {h}_{s_{i}}$ from~(\ref{alphy}) instead of~$\alpha_{i}$ in~$h_{i}$:
\begin{gather}
\tilde{h}_{i}=h_{i}|_{\alpha_{1}\rightarrow\tilde{h}_{s_{1}},\ldots,\alpha_{k}\rightarrow\tilde{h}_{s_{k}}},\qquad i=1,\dots,r,\qquad i\notin \{s_{1},\dots,s_{k}\}.\label{reszta}
\end{gather}

\begin{Definition}The functions $\tilde{h}_{i}=\tilde{h}_{i}(x,\tilde{\alpha}_{1},\dots ,\tilde{\alpha}_{k}),$ $i=1,\dots,r$, def\/ined through~(\ref{alphy}) and~(\ref{reszta}) are called the (generalized) St\"ackel transform of the functions~(\ref{hi}) with respect to the indices $\{s_{1},\dots,s_{k}\}$ (or with respect to the functions $h_{s_{1}},\ldots, h_{s_{k}}$).
\end{Definition}

Unless we extend the manifold $M$ this operation cannot be obtained by any coordinate change of variables. Moreover, if we perform again the St\"ackel transform on the functions $\tilde{h}_{i}$ with respect to $\tilde{h}_{s_{i}}$ we will receive back the functions~$h_{i}$ in~(\ref{hi}). In this sense the St\"ackel transform is a~reciprocal transform. Note also that neither~$r$ nor~$k$ are related to the dimension of the manifold~$M$.

In \cite{Blaszak-reciprocal} we proved that if $\dim M=2n$, $k=r=n$ and if all~$h_{i}$ are functionally independent then also all~$\tilde{h}_{i}$ will be functionally independent and if all $h_{i}~$are pairwise in involution with respect to $\Pi$ then also all~$\tilde{h}_{i}$ will pairwise Poisson-commute. That means that if the functions~$h_{i}$, $i=1,\ldots,n$ constitute a~Liouville integrable system then also $\tilde{h}_{i}$ will constitute a~Liouville integrable system. In other words, St\"ackel transform preserves Liouville integrability. But what about superintegrability?

\begin{Theorem}\label{wazny}Consider a maximally superintegrable system on a $2n$-dimensional Poisson mani\-fold, i.e., a~set of $2n-1$ functionally independent Hamiltonians $h_{1},\ldots,h_{2n-1}$ such that the first~$n$ Hamiltonians pairwise commute, and assume that all the Hamiltonians depend on $k\leq n$ parameters~$\alpha _{i}$:
\begin{gather}
h_{i} =h_{i}(x,\alpha_{1},\ldots,\alpha_{k}),\qquad i=1,\ldots,2n-1,\nonumber\\
 \{ h_{i},h_{j} \} =0,\qquad i,j=1,\ldots,n, \quad \text{for all} \ \alpha_{i},\label{supint}\\
 \{ h_{1},h_{n+j} \} =0,\qquad j=1,\ldots,n-1, \quad \text{for all} \ \alpha_{i}.\nonumber
\end{gather}
Suppose that $\{s_{1},\dots,s_{k}\}$ $\subset\{1,\dots,2n-1\}$ are chosen so that $s_{1}=1$ and that $\{s_{2},\dots,s_{k}\}$ $\subset\{2,\dots,n\}$ and moreover that $h_{1}=h_{1}(x,\alpha_{1})$. Then the St\"ackel transform $\tilde{h}_{i},$ $i=1,\ldots,$ $2n-1$ given by \eqref{alphy}, \eqref{reszta} also satisfy~\eqref{supint} and therefore constitute a maximally superintegrable system.
\end{Theorem}

Note that the Hamiltonian $h_{1}$ is now distinguished as the one that commutes with all the remaining~$h_{i}$ and as it can only depend on one
parameter. Note also that the f\/irst $n$ functions~$h_{i}$ pairwise commute with each other and therefore constitute a Liouville integrable system. The same is true about the f\/irst~$n$ functions~$\tilde{h}_{i}$.

\begin{proof}
Dif\/ferentiating the identity
\begin{gather*}
h_{s_{i}}\big(x,\tilde{h}_{s_{1}}(x,\tilde{\alpha}_{1},\dots ,\tilde{\alpha}_{k}),\dots ,\tilde{h}_{s_{k}}(x,\tilde{\alpha}_{1},\dots ,\tilde{\alpha} _{k})\big)=\tilde{\alpha}_{i},\qquad i=1,\dots ,k
\end{gather*}
with respect to $x$ we get
\begin{gather}
dh_{s_{i}}=-\sum_{j=1}^{k}\frac{\partial h_{s_{i}}}{\partial\alpha_{j}}d\tilde{h}_{s_{j}}, \qquad i=1,\ldots,k, \label{trdh1}
\end{gather}
while dif\/ferentiation of (\ref{reszta}) yields
\begin{gather}
dh_{_{i}}=d\tilde{h}_{i}-\sum_{j=1}^{k}\frac{\partial h_{i}}{\partial
\alpha_{j}}d\tilde{h}_{s_{j}},\qquad i=1,\ldots,2n-1,\qquad i\notin\{s_{1},\dots,s_{k}\}.\label{trdh2}
\end{gather}
The transformation (\ref{trdh1}), (\ref{trdh2}) can be written in a matrix form as
\begin{gather*}
dh=Ad\tilde{h},
\end{gather*}
where we denote $dh=(dh_{1},\ldots,dh_{2n-1})^{T}$ and $d\tilde{h}=(d\tilde {h}_{1},\ldots,d\tilde{h}_{2n-1})^{T}$ and where the $(2n-1)$ $\times(2n-1)$ matrix $A$ has the form
\begin{gather*}
A_{ij}=\delta_{ij}\qquad \text{for} \ \ j\notin\{s_{1},\dots,s_{k}\},\qquad A_{is_{j}}=-\frac{\partial h_{i}}{\partial\alpha_{j}}\qquad \text{for} \ \ j=1,\ldots,k.
\end{gather*}
Since
\begin{gather*}
\det A=\pm\det\left( \frac{\partial h_{s_{i}}}{\partial\alpha_{j}}\right)
\end{gather*}
is not zero and since $h_{i}$ are by assumption functionally independent on~$M$ we conclude that also the functions~$\tilde{h}_{i}$ are functionally independent on~$M$. Further, since $s_{k}\leq n$ (the St\"ackel transform is taken with respect to the Hamiltonians belonging to the Liouville integrable system $h_{1},\ldots,h_{n}$) the columns with derivatives of~$h_{i}$ with respect to parameters~$\alpha_{j}$ all lie in the left hand side of the matrix~$A$. Moreover, the fact that $h_{1}=h_{1}(x,\alpha_{1})$ also means that the f\/irst row of $A$ is zero except $A_{11}=-\frac{\partial h_{i}}{\partial \alpha_{1}}$. Let us now introduce the $(2n-1)\times(2n-1)$ matrices~$C$ and~$D$ through $C_{ij}= \{ h_{i},h_{j} \} $ and $D_{ij}= \{ \tilde{h}_{i},\tilde{h}_{j} \}$. A~direct calculation yields
\begin{gather*}
 \big\{ \tilde{h}_{i},\tilde{h}_{j} \big\} =\sum_{l_{1},l_{2}=1} ^{2n-1}\big( A^{-1}\big) _{il_{1}}\big( A^{-1}\big) _{jl_{2}} \{ h_{l_{1}},h_{l_{2}} \} _{\Pi}
\end{gather*}
or in matrix form
\begin{gather*}
D=A^{-1}C\big( A^{-1}\big) ^{T},
\end{gather*}
and due to the aforementioned structure of $A$ we have $D_{ij}=0$ for $i,j=1,\ldots, n$ (meaning that $\tilde{h}_{1},\ldots,\tilde{h}_{n}$ constitute
a Liouville integrable system) and moreover that $D_{1i}=D_{i1}=0$ for $i=1,\ldots,2n-1$, so that $\{ \tilde{h}_{1},\tilde{h}_{i}\} =0$ for all $i$. That concludes the proof.
\end{proof}

\begin{Remark} A similar statement with an analogous proof is valid for any superintegrable system of the form (\ref{supint}), not only the maximally superintegrable one.
\end{Remark}

\section[St\"ackel transform of maximally superintegrable St\"ackel systems]{St\"ackel transform of maximally superintegrable\\ St\"ackel systems}\label{section5}

In this section we perform those St\"ackel transforms of our systems (\ref{BenSC})--(\ref{laurent}) that preserve maximal superintegrability. According to Theorem~\ref{wazny}, the Hamiltonian~$h_{1}$ of the considered system can only depend on one parameter $h_{1}=h_{1}(x,\alpha)$. It is then natural to choose one of the~$a_{k}$ in~(\ref{laurent}) as this parameter.

Consider thus a maximally superintegrable system $(h_{1},\dots,h_{2n-1})$ with the f\/irst~$n$ commuting Hamiltonians $h_{1},\ldots,h_{n}$ def\/ined by our separation relations
\begin{gather*}
\sum_{s\in I}\alpha_{s}\lambda_{i}^{s}+h_{1}\lambda_{i}^{n-1}+h_{2}\lambda _{i}^{n-2}+\dots +h_{n}=\frac{1}{2}f(\lambda_{i})\mu_{i}^{2},\qquad i=1,\ldots,n,
\end{gather*}
where the index set $I$ satisf\/ies the assumptions of Theorem~\ref{super} and where the higher integrals $h_{n+r}$ are constructed as usual through $h_{n+r}= \{ h_{r+1},P \} $ with~$P$ constructed as in Theorem~\ref{super}. Let us now choose one of the parameters $\alpha_{s}$, with $s\in I$, say~$\alpha_{k}$, (we will suppose that $k\geq n$ or $k<0$ otherwise the corresponding potential is trivial, as explained earlier) and def\/ine the functions $H_{r}$, $r=1,\ldots,2n-1$, through
\begin{gather}
h_{r}=H_{r}+\alpha_{k}V_{r}^{(k)}, \qquad r=1,\ldots,2n-1.\label{ziuta}
\end{gather}
Then $V_{r}^{(k)}$ for $r=1,\ldots,n$ obviously coincide with $V_{r}^{(k)}$ def\/ined through (\ref{BenSC})--(\ref{laurent}) or equivalently through~(\ref{U}), (\ref{R}).

We now perform the St\"ackel transform on this system $(h_{1},\dots,h_{2n-1})$ with respect to the chosen parameter~$\alpha_{k}$ as described in Theorem~\ref{wazny}. It means that we f\/irst solve the relation $h_{1}=\tilde{\alpha}$, i.e., $H_{1}+\alpha_{k}V_{1}^{(k)}=\tilde{\alpha}$ with respect to~$\alpha_{k}$ which yields
\begin{gather}
\tilde{h}_{1}=\alpha_{k}=-\frac{1}{V_{1}^{(k)}}H_{1}+\tilde{\alpha}\frac{1}{V_{1}^{(k)}},\label{basia}
\end{gather}
and then replace $\alpha_{k}$ with $\tilde{h}_{1}$ in all the remaining Hamiltonians:
\begin{gather}
\tilde{h}_{r}=H_{r}-\frac{V_{r}^{(k)}}{V_{1}^{(k)}}H_{1}+\tilde{\alpha}\frac{V_{r}^{(k)}}{V_{1}^{(k)}},\qquad r=2,\dots,2n-1.\label{4.1}
\end{gather}
We obtain in this way a new superintegrable system $(\tilde{h}_{1},\dots,\tilde{h}_{2n-1})$ \ where the f\/irst~$n$ commuting Hamiltonians
$\tilde{h}_{r}$ are def\/ined by (see~\cite{Blaszak-reciprocal}) the following separation relations
\begin{gather}
\tilde{h}_{1}\lambda_{i}^{k}+\sum_{s\in I,\, s\neq k}\alpha_{s}\lambda_{i}^{s}+\tilde{\alpha}\lambda_{i}^{n-1}+\tilde{h}_{2}\lambda_{i}
^{n-2}+\dots +\tilde{h}_{n}=\frac{1}{2}f(\lambda_{i})\mu_{i}^{2}, \qquad i=1,\ldots,n,\label{4.2}
\end{gather}
as it is easy to see, since on the level of the separation relations our St\"{a}ckel transform replaces~$\alpha _{k}$ with $\widetilde{h}_{1}$ and $h_{1}$ with~$\widetilde{\alpha }$. For $k\geq n$ or $k<-1$ the system~(\ref{4.2}) is no longer in the class~(\ref{BenSC}), while for $k=-1$ it can be easily transformed by a~simple point transformation to the form~(\ref{BenSC}).

\begin{Lemma}\label{lemat}The separable system
\begin{gather*}
\alpha_{k}\lambda_{i}^{k}+\sum_{s\in I,\, s\neq k}\alpha_{s}\lambda_{i}^{s}+h_{1}\lambda_{i}^{n-1}+h_{2}\lambda_{i}^{n-2}+\dots +h_{n}=\frac{1}{2}\lambda_{i}^{m}\mu_{i}^{2}, \qquad i=1,\ldots,n
\end{gather*}
attains after the St\"ackel transform \eqref{basia}, \eqref{4.1} and after the consecutive point transformation on~$M$ given by
\begin{gather}
\lambda_{i}\rightarrow1/\lambda_{i},\qquad \mu_{i}\rightarrow-\lambda_{i}^{2}\mu_{i},\qquad i=1,\ldots,n \label{tran}
\end{gather}
the form%
\begin{gather}
\tilde{\alpha}\lambda_{i}^{-1}+\sum_{s\in I,\, s\neq k}\alpha_{s}\lambda_i^{n-2-s}+\tilde{h}_{1}\lambda_{i}^{n-k-2}+\tilde{h}_{n}\lambda
_{i}^{n-2}+\dots +\tilde{h}_{2}\nonumber\\
\qquad =\frac{1}{2}\lambda_{i}^{n-m+2}\mu_{i}^{2}, \qquad i=1,\ldots,n.\label{tran1}
\end{gather}
\end{Lemma}

Note that the transformation (\ref{tran}) on $M$ does not change the separation web of the system on~$Q$. Denoting, as before
\begin{gather}
\tilde{h}_{r}=\widetilde{H}_{r}+\tilde{\alpha}\tilde{V}_{r}, \qquad r=1,\ldots,2n-1,\label{4.1a}
\end{gather}
where~$\tilde{h}_{r}$ for $r=1,\ldots ,n$ are def\/ined by~(\ref{4.2}) while $\tilde{h}_{r}$ for $r=n+1,\ldots ,2n-1$ are obtained as usual through
$\tilde{h}_{n+r}=\big\{ \tilde{h}_{r+1},P\big\}$, we see from (\ref{4.1}) that
\begin{gather*}
\tilde{V}_{r}=V_{r}-\frac{V_{r}^{(k)}}{V_{1}^{(k)}}V_{1}, \qquad r=2,\dots,2n-1,
\end{gather*}
and from (\ref{basia}) it also follows that the geodesic part $\tilde{E}_{1}$
of $\tilde{h}_{1}$ has the form
\begin{gather}
\tilde{E}_{1}=\sum_{i,j=1}^{n}\tilde{G}^{ij}p_{i}p_{j},\qquad \tilde{G}=-\frac {1}{V_{1}^{(k)}}G.\label{Gbar}
\end{gather}
It means that the metric $\tilde{G}$ is a \emph{conformal deformation} of either a f\/lat or a constant curvature metric~$G$. In the following theorem we list the cases when the metric $\tilde{G}$ is actually f\/lat or of constant curvature as well. The theorem is formulated only for~$f$ in~(\ref{fm}) being a~monomial, \thinspace$f=\lambda^{m}$ (in this case there is a maximum number of f\/lat metrics~$\tilde{G}$).

\begin{Theorem} \label{wazny2}
Consider the system~\eqref{4.2} with $f=\lambda ^{m}$ where $m\in \{ 0,\ldots ,n+1 \}$.
\begin{enumerate}\itemsep=0pt
\item[$(i)$] For $0\leq m\leq n-1$ the system \eqref{4.2} is maximally superintegrable for $k\in\{ -m,\ldots,-1,n,$ $\ldots,2n-m-1\} $. The metric~$\tilde{G}$ in~\eqref{Gbar} is flat for $k\in \{ - [m/2 ] ,\ldots,-1,n,\ldots,n-1+ [ (n-m)/2 ]\}$, where $[ \cdot ]$ denotes the integer part. Moreover, for $m=1$ and $k=-1$~$\tilde{G}$ is of constant curvature. Otherwise $\tilde{G}$ is conformally flat.

\item[$(ii)$] For $m=n$ the system \eqref{4.2} is maximally superintegrable for $k\in \{ -(n-2),\ldots,-1,n\} $. The metric~$\tilde{G}$ in~\eqref{Gbar} is flat for $k\in\{ -[ n/2] ,\ldots ,-1\} $. Otherwise $\tilde{G}$ is conformally flat.

\item[$(iii)$] For $m=n+1$ the system (\ref{4.2}) is maximally superintegrable for $k\in \{-(n-1),\ldots,-1\} $. The metric $\tilde{G}$ in~\eqref{Gbar} is flat for $k\in \{ -[ (n+1)/2] ,\ldots ,-1\} $. Otherwise $\tilde{G}$ is conformally flat.
\end{enumerate}
\end{Theorem}

If $f$ is a polynomial then the admissible values of $k$ must satisfy the above type of bonds for all powers of $\lambda$ in~$f$, not only for the highest power $m$ so we choose not to present this more general theorem, only to maintain the simplicity of the picture. In order to prove Theorem~\ref{wazny2} we need one more lemma.

\begin{Lemma}[\cite{Blaszak-Schouten}] The Ricci scalars $R$ and $\tilde{R}$ of the conformally related $($covariant$)$ metric tensors~$g$ and $\tilde{g}=\sigma g$ are related through
\begin{gather}
\tilde{R}=\sigma^{-1}R-\frac{1}{2}(n-1)\sigma^{-1}s_{ij}G^{ij},\label{Kt}
\end{gather}
where $G=g^{-1}$ and where
\begin{gather*}
s_{ij}=s_{ji}=2\nabla_{i}s_{j}-s_{i}s_{j}+\frac{1}{2}g_{ij}s_{k} s^{k}\qquad \text{with} \quad s_{i}=\sigma^{-1}\frac{\partial\sigma}{\partial x_{i}},
\end{gather*}
where $x_{i}$ are any coordinates on the manifold.
\end{Lemma}

\begin{proof}[Proof of Theorem \ref{wazny2}] The values of $k$ for which~(\ref{ziuta}) is maximally superintegrable follows from the specif\/ication of Theorem~\ref{super} to the case $f=\lambda^{m}$. For $(i)$ and $(ii)$ the metric $G$ in~(\ref{Gbar}) is f\/lat so that its Ricci scalar $R=0$. Therefore, according to (\ref{Kt}), $\tilde{R}=0$ if and only if $s_{ij}G^{ij}=0$. This condition can be ef\/fectively calculated in f\/lat coordinates $r_{i}$ of the metric~$G$ given by~\cite{Blaszak-BS}
\begin{gather*}
\begin{split}
&q_{i} =r_{i}+\frac{1}{4}\sum\limits_{j=1}^{i-1}r_{j}r_{i-j}, \qquad i=1,\ldots,n-m,\\
&q_{i} =-\frac{1}{4}\sum\limits_{j=i}^{n} r_{j}r_{n-j+i}, \qquad i=n-m+1,\ldots,m.
\end{split}
\end{gather*}
In these coordinates
\begin{gather*}
(G_{m}) ^{kl}=\delta_{n-m+1}^{k+l}+\delta_{2n-m+1}^{k+l}
\end{gather*}
and the condition $s_{ij}G^{ij}=0$ yields both statements. The case in~$(i)$ when $\tilde{G}$ is of constant curvature ($m=1,k=-1$) can be however more ef\/fectively proven using Lemma~\ref{lemat} since in this case the system~\eqref{4.2} attains after the transformation~(\ref{tran}) the form
\begin{gather*}
\tilde{h}_{1}\lambda_{i}^{n-1}+\sum_{s\in I,\, s\neq k}\alpha_{s}\lambda_i^{n-2-s}+\tilde{h}_{n}\lambda_{i}^{n-2}+\dots +\tilde{h}_{2}+\tilde
{\alpha}\lambda_{i}^{-1}=\frac{1}{2}\lambda_{i}^{n+1}\mu_{i}^{2}, \qquad i=1,\ldots,n.
\end{gather*}
Due to Remark \ref{remark} the metric $\tilde{G}$ of this system has constant curvature. Finally, in the case~$(iii)$ ($m=n+1$) we have only negative potentials so by using Lemma~\ref{lemat} we transform this system to
\begin{gather*}
\tilde{h}_{1}\lambda_{i}^{n-k-2}+\sum_{s\in I,\, s\neq k}\alpha _{s}\lambda_i^{n-2-s}+\tilde{h}_{n}\lambda_{i}^{n-2}+\dots +\tilde{h}_{2} +\tilde{\alpha}\lambda_{i}^{-1}=\frac{1}{2}\lambda_{i}\mu_{i}^{2}, \qquad i=1,\ldots,n,
\end{gather*}
where $k<0$, and this is the system from case $(i)$ with $m=1$ and therefore $\tilde{G}$ is f\/lat for $k\geq-[ (n+1)/2] $. For other values of~$k$ the metric $\tilde{G}$ is conformally f\/lat.
\end{proof}

If $Y(V_{1}^{(k)})=0$ then $Y(1/V_{1}^{(k)})=0$ and due to (\ref{Gbar}) also $L_{Y}\tilde{G}=0$ so that $\{ \tilde{h}_{1},P\} =0$ as well and the same $P$ as in the ``non-tilde''-case (i.e., before the St\"ackel transform) can be used as an alternative def\/inition of extra Hamiltonians through $\bar{h}_{n+r}= \{ \tilde{h}_{r+1},P \}$, $r=1,\ldots,n-1$. This is however no longer true if $Y\big(V_{1}^{(k)}\big)=c\neq0$ (according to Theorem~\ref{super}, it happens only in the case when $m<n$ and $k=2n-m-1$). It turns out that it leads to the same extra integrals of motion, as the following theorem states

\begin{Theorem}\label{tosamo} If $Y(V_{1}^{(k)})=0$ then both sets of extra integrals of motion:
\begin{gather*}
\bar{h}_{n+r}= \{ \tilde{h}_{r+1},P \}, \qquad r=1,\ldots,n-1
\end{gather*}
and
\begin{gather*}
\tilde{h}_{n+r}= h_{n+r} \vert _{\alpha=\widetilde{h}_{1}(\widetilde{\alpha})}, \qquad r=1,\ldots,n-1
\end{gather*}
coincide.
\end{Theorem}

\begin{proof}
On one hand, according to (\ref{4.1}) and due to the fact that $\{\tilde{h}_{1},P\} =0$ we have
\begin{gather*}
\begin{split} &
\bar{h}_{n+r} =\big\{ \tilde{h}_{r+1},P\big\} =\left\{ H_{r+1}-\frac{V_{r+1}^{(k)}}{V_{1}^{(k)}}H_{1}+\tilde{\alpha}\frac{V_{r+1}^{(k)}
}{V_{1}^{(k)}},P\right\} \\
& \hphantom{\bar{h}_{n+r}}{} = \{ H_{r+1},P \} -\frac{H_{1}}{V_{1}^{(k)}}\big\{ V_{r+1}^{(k)},P\big\} +\frac{\tilde{\alpha}}{V_{1}^{(k)}}\big\{ V_{r+1}^{(k)},P\big\} = \{ H_{r+1},P\} +\tilde{h}_{1}\big\{ V_{r+1}^{(k)},P\big\}.
\end{split}
\end{gather*}
On the other hand, due to
\begin{gather*}
\tilde{h}_{n+r}=\left. h_{n+r}\right\vert _{\alpha=\widetilde{h}_{1}(\widetilde{\alpha})}= \{ h_{r+1},P \} \vert
_{\alpha=\widetilde{h}_{1}(\widetilde{\alpha})}= \{ H_{r+1},P \}
+\alpha\big\{ V_{r+1}^{(k)},P\big\} \big\vert _{\alpha =\widetilde{h}_{1}(\widetilde{\alpha})},
\end{gather*}
which yields the same result.
\end{proof}

Thus, if $Y(V_{1}^{(k)})=0$, the diagram below commutes%
\begin{gather*}
\begin{matrix}
\left( h_{1},\ldots,h_{n}\right) & \overset{P}{\longrightarrow} & \left(
h_{1},\ldots,h_{2n-1}\right) \ \text{with} \ h_{n+r}=\left\{ h_{r+1},P\right\}
\\
\mid & & \mid\\
\text{St\"ackel transform} & & \text{St\"ackel transform}\\
\downarrow & & \downarrow\\
\big( \tilde{h}_{1},\ldots,\tilde{h}_{n}\big) & \overset{P}
{\longrightarrow} & \big( \tilde{h}_{1},\ldots,\tilde{h}_{2n-1}\big) \ \text{with} \ \tilde{h}_{n+r}=\big\{ \tilde{h}_{r+1},P\big\}.
\end{matrix}
\end{gather*}

\begin{Example}\label{4}Let us apply the relations (\ref{basia}), (\ref{4.1}) to perform the St\"ackel transform on the system from Example~\ref{2}. To keep the formulas simple, we assume that all the~$\alpha_{s}$ in~(\ref{laurent}) are zero except the transformation parameter~$\alpha_{k}$. Thus, we consider again the system given by the separation relations
\begin{gather*}
\alpha_{k}\lambda_{i}^{k}+h_{1}\lambda_{i}^{2}+h_{2}\lambda_{i}+h_{3}=\frac
{1}{2}\lambda_{i}\mu_{i}^{2}, \qquad i=1,2,3
\end{gather*}
with $k=-1$, $3$ or $4$, respectively. Applying St\"ackel transform to the resulting Hamiltonians (\ref{ham2})--(\ref{jola1}) we obtain a maximally superintegrable system with the separation relations of the form:
\begin{gather}
\tilde{h}_{1}\lambda_{i}^{k}+\tilde{\alpha}\lambda_{i}^{2}+\tilde{h}_{2}\lambda_{i}+\tilde{h}_{3}=\frac{1}{2}\lambda_{i}\mu_{i}^{2}, \qquad i=1,2,3.\label{sep1}
\end{gather}
Again we perform our calculations in the $(r,s)$-variables (\ref{r}), (\ref{s}). Explicitly, we obtain for $k=-1$
\begin{gather}
\tilde{h}_{1} =\frac{1}{8}r_{3}^{2}s_{3}^{2}+\frac{1}{4}r_{3}^{2}s_{1}s_{2}-\frac{1}{4}\tilde{\alpha}r_{3}^{2},\nonumber\\
\tilde{h}_{2} =\frac{1}{2}s_{1}^{2}-\frac{1}{2}r_{2}s_{2}^{2}-\frac{1}{2}r_{1}s_{1}s_{2}-\frac{1}{2}r_{3}s_{2}s_{3}+\tilde{\alpha}r_{1},\nonumber\\
\tilde{h}_{3} =\frac{1}{8}r_{3}^{2}s_{2}^{2}-\left( \frac{1}{4}r_{1}^{2}+r_{2}\right) s_{1}s_{2}-\frac{1}{2}r_{3}s_{1}s_{3}-\frac{1}{4}r_{1}
r_{3}s_{2}s_{3}+\frac{1}{4}\tilde{\alpha}\big( r_{1}^{2}+4r_{2}\big),\nonumber\\
\tilde{h}_{4} =-\frac{1}{2}s_{2}^{2}, \qquad \tilde{h}_{5} =-s_{1}s_{2}+\tilde{\alpha},\label{s1}
\end{gather}
for $k=3$
\begin{gather*}
\tilde{h}_{1} =-\frac{1}{r_{1}}s_{1}s_{2}-\frac{1}{2}\frac{1}{r_{1}}s_{3}^{2}+\tilde{\alpha}\frac{1}{r_{1}},\\
\tilde{h}_{2} =\frac{1}{2}s_{1}^{2}+\frac{1}{4}\frac{r_{1}^{2}-4r_{2}}{r_{1}}s_{1}s_{2}-\frac{1}{2}r_{2}s_{2}^{2}-\frac{1}{2}r_{3}s_{2}s_{3}
+\frac{1}{8}\frac{3r_{1}^{2}-4r_{2}}{r_{1}}s_{3}^{2}+\frac{1}{4}\tilde{\alpha }\frac{r_{1}^{2}+4r_{2}}{r_{1}},\\
\tilde{h}_{3} =\frac{1}{4}\frac{r_{3}^{2}}{r_{1}}s_{1}s_{2}-\frac{1}{2}r_{3}s_{1}s_{3}+\frac{1}{8}r_{3}^{2}s_{2}^{2}-\frac{1}{4}r_{1}r_{3}
s_{2}s_{3}+\frac{1}{8}\frac{r_{1}^{3}+4r_{1}r_{2}+r_{3}^{2}}{r_{1}}s_{3}^{2}-\frac{1}{4}\tilde{\alpha}\frac{r_{3}^{2}}{r_{1}},\\
\tilde{h}_{4} =-\frac{1}{r_{1}}s_{1}s_{2}-\frac{1}{2}s_{2}^{2}-\frac{1}{2}\frac{1}{r_{1}}s_{3}^{2}+\tilde{\alpha}\frac{1}{r_{1}},\qquad
\tilde{h}_{5} =\frac{1}{2}s_{3}^{2},
\end{gather*}
and for $k=4$
\begin{gather}
\tilde{h}_{1} =-\frac{1}{r_{2}-\frac{3}{4}r_{1}^{2}}s_{1}s_{2}-\frac{1}{2}\frac{1}{r_{2}-\frac{3}{4}r_{1}^{2}}s_{3}^{2}+\tilde{\alpha}\frac{1}
{r_{2}-\frac{3}{4}r_{1}^{2}},\nonumber\\
\tilde{h}_{2} =\frac{1}{2}s_{1}^{2}-\frac{1}{2}r_{2}s_{2}^{2}-\frac{1}{8}\frac{2r_{1}^{3}-8r_{1}r_{2}-r_{3}^{2}}{r_{2}-\frac{3}{4}r_{1}^{2}}
s_{3}^{2}-\frac{1}{8}\frac{r_{1}^{3}-12r_{1}r_{2}-2r_{3}^{2}}{r_{2}-\frac{3}{4}r_{1}^{2}}s_{1}s_{2}-\frac{1}{2}r_{3}s_{2}s_{3}\nonumber\\
\hphantom{\tilde{h}_{2} =}{} -\tilde{\alpha}\frac{r_{1}r_{2}+\frac{1}{4}r_{1}^{3}+\frac{1}{4}r_{3}^{2}}{r_{2}-\frac{3}{4}r_{1}^{2}},\nonumber\\
\tilde{h}_{3} =\frac{1}{8}r_{3}^{2}s_{2}^{2}-\frac{1}{32}\frac{3r_{1}^{4}+8r_{1}^{2}r_{2}+4r_{1}r_{3}^{2}+16r_{2}^{2}}{r_{2}-\frac{3}{4}r_{1}^{2}
}s_{3}^{2}+\frac{1}{4}\frac{r_{1}r_{3}^{2}}{r_{2}-\frac{3}{4}r_{1}^{2}}s_{1}s_{2}-\frac{1}{2}r_{3}s_{1}s_{3}-\frac{1}{4}r_{1}r_{3}s_{2}s_{3}\nonumber\\
\hphantom{\tilde{h}_{3} =}{} +\frac{1}{4}\tilde{\alpha}\frac{r_{1}r_{3}^{2}}{r_{2}-\frac{3}{4}r_{1}^{2}},\nonumber\\
\tilde{h}_{4} =-\frac{1}{2}s_{2}^{2}+\frac{1}{2}\frac{r_{1}}{r_{2}-\frac{3}{4}r_{1}^{2}}s_{3}^{2}+\frac{r_{1}}{r_{2}-\frac{3}{4}r_{1}^{2}}s_{1}
s_{2}-\tilde{\alpha}\frac{r_{1}}{r_{2}-\frac{3}{4}r_{1}^{2}}, \qquad \tilde{h}_{5} =\frac{1}{2}s_{3}^{2}.\label{s3}
\end{gather}
According to part $(i)$ of Theorem~\ref{wazny2} the metrics of $\tilde{h}_{1} $ are of constant curvature, f\/lat and conformally f\/lat, respectively.
\end{Example}

\section{Quantization of maximally superintegrable St\"ackel systems}\label{section6}

This section is devoted to separable quantizations of St\"ackel systems\ that were considered in the classical setting in the previous sections. Let us consider, as in the classical case, an $n$-dimensional Riemannian space $Q $ equipped with a matric tensor $g$ and the quadratic in momenta Hamiltonian on the cotangent bundle $T^{\ast}Q$:
\begin{gather*}
h=\frac{1}{2}\sum_{i,j=1}^{n}p_{i}A^{ij}(x)p_{j}+U(x).
\end{gather*}
By its \emph{minimal quantization}~\cite{Blaszak-inpress} we mean the following self-adjoint operator
\begin{gather}
\widehat{h}=-\frac{1}{2}\hslash^{2}\sum_{i,j=1}^{n}\nabla_{i}A^{ij}%
(x)\nabla_{j}+U(x)=-\frac{1}{2}\hslash^{2}\sum_{i,j=1}^{n}\frac{1}%
{\sqrt{\left\vert g\right\vert }}\partial_{i}\sqrt{\left\vert g\right\vert
}A^{ij}(x)\partial_{j}+U(x)\label{hq}
\end{gather}
(both expressions on the right hand side of (\ref{hq}) are equivalent) acting in the Hilbert space
\begin{gather*}
\mathcal{H}=L^{2}(Q,d\mu), \qquad d\mu= \vert g \vert^{1/2}dx, \qquad \vert g \vert = \vert \det g \vert,
\end{gather*}
where $\nabla$ is the Levi-Civita connection of the metric $g$. Note that a~priori there is no relation between the tensor~$A$ and the metric~$g$. Let us now consider an arbitrary St\"ackel system of the form~(\ref{Benham}) coming from the separation relations~(\ref{sr}). Applying the procedure of minimal quantization to this system will in general yield a non-integrable and non-separable quantum system. In order to preserve integrability and separability we have to carefully choose the metric~$g$. To do this, we will use the following theorem, proved in~\cite{Blaszak-inpress}.

\begin{Theorem}\label{wazny3}Suppose that $h_{j}$ are Hamiltonian functions \eqref{Benham}, defined by separation relations~\eqref{sr}. Suppose also that $\theta$ is an arbitrary function of one variable. Applying to~$h_{j}$ the procedure of minimal quantization~\eqref{hq} with the metric tensor
\begin{gather}
g=\varphi^{\frac{2}{n}}g_{\theta}, \label{metryka}
\end{gather}
where $g_{\theta}=G_{\theta}^{-1}$ with $G_{\theta}$ given by
\begin{gather}
G_{\theta}=\operatorname{diag}\left( \frac{\theta(\lambda_{1})}{\Delta_{1}},\ldots,\frac{\theta(\lambda_{n})}{\Delta_{n}}\right), \label{theta}
\end{gather}
and with $\varphi$ being a particular function of $\lambda_{1},\dots,\lambda _{n}$, uniquely defined by~\eqref{sr} $($see formula~{\rm (27)} in~{\rm \cite{Blaszak-inpress}} for details$)$, we obtain a quantum integrable and separable system. More precisely, we obtain $n$ operators $\widehat{h}_{i}$ of the form \eqref{hq} such that~$(i)$ $[ \widehat{h}_{i},\widehat{h}_{j}] =0$ for all $i$, $j$ and~$(ii)$ eigenvalue problems for all~$\widehat{h}_{i}$
\begin{gather*}
\widehat{h}_{i}\Psi=\varepsilon_{i}\Psi, \qquad i=1,\ldots,n
\end{gather*}
have for each choice of eigenvalues $\varepsilon_{i}$ of $\widehat{h}_{i}$ the common multiplicatively separable eigenfunction $\Psi(\lambda_{1},\dots,\lambda_{n})=\prod\limits_{i=1}^{n}\psi(\lambda_{i})$ with $\psi$ satisfying the following ODE $($quantum separation relation$)$
\begin{gather}
\big(\varepsilon_{1}\lambda^{\gamma_{1}}+\varepsilon_{2}\lambda^{\gamma_{2} }+\dotsb+\varepsilon_{n}\big)\psi(\lambda) \nonumber\\
\qquad{} =-\frac{1}{2}\hbar^{2}f(\lambda)\left[\frac{d^{2}\psi(\lambda)}{d\lambda^{2}}+\left( \frac{f^{\prime}(\lambda
)}{f(\lambda)}-\frac{1}{2}\frac{\theta^{\prime}(\lambda)}{\theta(\lambda
)}\right) \frac{d\psi(\lambda)}{d\lambda}\right] +\sigma(\lambda)\psi(\lambda). \label{QS}
\end{gather}
\end{Theorem}

\begin{Remark}\label{jakkwantowac}For St\"ackel systems def\/ined by \eqref{BenSC}, when $(\gamma_{1},\dots,\gamma_{n})=(n-1,\dots,0)$ in~(\ref{sr}), we have $\varphi=1$ and the most natural choice in~(\ref{theta}) is to put $\theta=f$ which yields the metric for quantization
\begin{gather}
G=G_{f}=A_{1}.\label{MQB}
\end{gather}
On the other hand, for St\"ackel systems def\/ined by~(\ref{4.2}) we have $\varphi=-V_{1}^{(k)}$ (as it follows from~(\ref{basia}) and the formula~(27) in~\cite{Blaszak-inpress}) and again the simplest choice in~(\ref{theta}) is to put $\theta=f$ which yields according to~(\ref{metryka}) and~(\ref{Gbar}) the metric for quantization
\begin{gather}
G=\varphi^{-\frac{2}{n}}G_{f}=\varphi^{1-\frac{2}{n}}\tilde{A}_{1}.\label{MQ}
\end{gather}
\end{Remark}

For the choice (\ref{MQB}) and (\ref{MQ}) the quantum separation equation~(\ref{QS}) reduce to
\begin{gather}
\big(\varepsilon_{1}\lambda^{\gamma_{1}}+\varepsilon_{2}\lambda^{\gamma_{2}
}+\dotsb+\varepsilon_{n}\big)\psi(\lambda)=-\frac{1}{2}\hbar^{2}f(\lambda)\left[
\frac{d^{2}\psi(\lambda)}{d\lambda^{2}}+\frac{1}{2}\frac{f^{\prime}(\lambda
)}{f(\lambda)}\frac{d\psi(\lambda)}{d\lambda}\right] +\sigma(\lambda)\psi(\lambda),\label{sepe}
\end{gather}
where $( \gamma_{1},\ldots, \gamma _{n}) =(n-1,n-2,\ldots ,0)$ in the f\/irst case~\eqref{MQB} and $( \gamma _{1},\ldots, \gamma _{n}) =(k,n-2,n-3$, $\ldots ,0)$ in the second case~\eqref{MQ}.

Let us now pass to the issue of quantum superintegrability of considered St\"ackel systems. We formulate now a quantum analogue of Lemma~\ref{lemacik}.

\begin{Lemma}\label{lemacik2} Suppose that $\widehat{h}$ is given by \eqref{hq} and that $Y=\sum\limits_{i=1}^{n} y^{i}(x)\nabla_{i}$ is a vector field on the Riemannian manifold~$Q$ with a~metric~$g$. Then
\begin{gather*}
\big[ \widehat{h},Y\big] =\frac{1}{2}\hslash^{2}\sum_{i,j=1}^{n}\nabla_{i}( L_{Y}A) ^{ij}\nabla_{j}+\frac{1}{2}\hslash^{2}
\sum_{i,j,k=1}^{n}A^{ij}\big( \nabla_{j}\nabla_{k}y^{k}\big) \nabla_{i}-Y(U).
\end{gather*}
\end{Lemma}

One proves this lemma by a direct computation. Thus, a suf\/f\/icient condition for $[ \widehat{h},Y] =c$ is satisf\/ied when~$Y$ is a Killing vector for both $A$ and $g$ and if moreover $U$ is constant along~$Y$, that is when
\begin{gather}
L_{Y}A=0, \qquad L_{Y}g=0, \qquad Y(U)=c\label{wk}
\end{gather}
(note that $L_{Y}g=0$ implies $\sum\limits_{i=1}^{n}\nabla_{k}y^{k}=0$).

\begin{Corollary}\label{qint}Suppose we have a quantum integrable system on the configuration space $Q$, that is a set of $n$ commuting and algebraically independent operators $\widehat{h}_{1},\ldots,\widehat{h}_{n}$ of the form~\eqref{hq} acting in the Hilbert space $L^{2}(Q,\vert g\vert ^{1/2}dx)$ where $g$ is some metric on~$Q$. Suppose also that a~vector field~$Y$ satisfies~\eqref{wk} with~$A_{1}$ and~$U_{1}$ instead of~$A$ and~$U$ $($so that $[\widehat{h}_{1},Y] =c)$. Then, analogously to the classical case, the operators
\begin{gather}
\widehat{h}_{n+r}= \big[ \widehat{h}_{r+1},Y \big] =\frac{1}{2}\hslash ^{2}\sum_{i,j=1}^{n}\nabla_{i} ( L_{Y}A_{r+1} ) ^{ij}\nabla _{j}-Y(U_{r+1}), \qquad r=1,\ldots,n-1\label{op}
\end{gather}
satisfy $[ \widehat{h}_{n+r},\widehat{h}_{1}] =0$ and the system $\widehat{h}_{1},\ldots,\widehat{h}_{2n-1}$ is algebraically independent; that
is we obtain a quantum separable and quantum superintegrable system.
\end{Corollary}

We can now apply this corollary to construct quantum superintegrable counterparts of classical systems considered in previous sections. According
to Remark~\ref{jakkwantowac}, for the systems generated by the separation relations~(\ref{BenSC}) the most natural choice of the metric $g$ is to take $G=A_{1}$ as in~(\ref{MQB}). Then, by construction, $ [ \widehat{h}_{i},\widehat{h}_{j} ] =0$ for $i,j=1,\ldots,n$ while the remaining
opera\-tors~$\widehat{h}_{n+r}$ are constructed by the formula~(\ref{op}) and are~-- up to a sign~-- identical with minimal quantization (in the metric~$G$) of the extra integrals~$h_{n+r}$ obtained in~(\ref{struk}).

\begin{Example} Consider again separation relations~(\ref{sep}) from Example~\ref{2}, so that $f(\lambda)=\lambda$ and $\sigma=\alpha_{-1}\lambda^{-1} +\alpha_{3}\lambda ^{3}+\alpha_{4}\lambda^{4}$. Performing the minimal quantization of the Hamiltonians~(\ref{ham2}) in the metric $G=A_{1}$, i.e., given by~(\ref{Go}), we obtain, in the f\/lat $r$-coordinates~(\ref{r})
\begin{gather*}
\widehat{h}_{1} =-\frac{1}{2}\hslash^{2}\left( \partial_{1}\partial
_{2}+\frac{1}{2}\partial_{3}^{2}\right) +\alpha_{-1}V_{1}^{(-1)}(r)+\alpha_{3}V_{1}^{(3)}(r)+\alpha_{4}V_{1}^{(4)}(r),\\
\widehat{h}_{2} =-\frac{1}{4}\hslash^{2}\left( \partial_{1}^{2}-\partial_{2}r_{2}\partial_{2}+r_{1}\partial_{3}\partial_{3} +\frac{1}{2}\partial_{1}r_{1}\partial_{2}+\frac{1}{2}r_{1}\partial_{2}\partial
_{1}-r_{3}\frac{1}{2}\partial_{2}\partial_{3}-\frac{1}{2}\partial_{3}r_{3}\partial_{2}\right) \\
\hphantom{\widehat{h}_{2} =}{} +\alpha_{-1}V_{2}^{(-1)}(r)+\alpha_{3}V_{2}^{(3)}(r)+\alpha_{4}V_{2}^{(4)}(r),\\
\widehat{h}_{3} =-\frac{1}{8}\hslash^{2}\left( \frac{1}{2}r_{3}^{2}\partial_{2}^{2}+\left( 2r_{2}+\frac{1}{2}r_{1}^{2}\right) \partial
_{3}^{2}-r_{3}\partial_{1}\partial_{3}-\partial_{3}r_{3}\partial_{1}-\frac
{1}{2}r_{1}r_{3}\partial_{2}\partial_{3}-\frac{1}{2}r_{1}\partial_{3}r_{3}\partial_{2}\right) \\
\hphantom{\widehat{h}_{3} =}{} +\alpha_{-1}V_{3}^{(-1)}(r)+\alpha_{3}V_{3}^{(3)}(r)+\alpha_{4}V_{3}^{(4)}(r),
\end{gather*}
where $\partial_{i}=\partial/\partial r_{i}$ and $V_{i}^{(k)}$ are given by (\ref{pot1})--(\ref{pot3}). The respective separation equation, according to~(\ref{sep}) and~(\ref{sepe}), is of the form
\begin{gather*}
\big(\alpha_{-1}\lambda^{-1}+\alpha_{3}\lambda^{3}+\alpha_{4}\lambda^{4} +\varepsilon_{1}\lambda^{2}+\varepsilon_{2}\lambda+\varepsilon_{3}\big)\psi(\lambda) =-\frac{1}{2}\hbar^{2}\left[ \lambda\frac{d^{2}\psi(\lambda
)}{d\lambda^{2}}+\frac{1}{2}\frac{d\psi(\lambda)}{d\lambda}\right].
\end{gather*}
Now $Y=\partial_{2}$ satisf\/ies the conditions (\ref{wk}) and the extra operators $\widehat{h}_{4}$, $\widehat{h}_{5}$ can be obtained either by using the formula~(\ref{op}) or directly by minimal quantization of functions $h_{4}$, $h_{5}$ in~(\ref{jola1}). The result is (up to a sign)
\begin{gather*}
\widehat{h}_{4}=\frac{1}{4}\hslash^{2}\partial_{2}^{2}-\alpha_{3}+\alpha_{4}r_{1},\qquad \widehat{h}_{5}=-\frac{1}{4}\hslash^{2}\partial_{3}^{2}+\frac{4\alpha_{-1}}{r_{3}^{2}}.
\end{gather*}
\end{Example}

If we want to perform the separable quantization of superintegrable systems obtained by the St\"ackel transform, as in Section~\ref{section5}, we have two cases: either the system~-- after the St\"ackel transform~-- belongs again to the same class~(\ref{BenSC}) or belongs to the other class, given by the separation relations (\ref{4.2}) that are dif\/ferent from~(\ref{BenSC}) as soon as $k\neq-1$. Again by Remark~\ref{jakkwantowac}, in the f\/irst case the natural choice of the metric in which we perform the minimal quantization is to take $\tilde{G}=\widetilde{A}_{1}$, i.e., $\tilde{G}$ as given by~(\ref{Gbar}). In the second case we have to use the metric given by (\ref{metryka}) which in our case is given by~(\ref{MQ}), i.e., by $G=\varphi^{1-\frac{2}{n}}\tilde{A}_{1}$ with $\varphi=-V_{1}^{(k)}$.

\begin{Example} Let us now minimally quantize the St\"ackel Hamiltonians $\tilde{h}_{1}$, $\tilde{h}_{2}$, $\tilde{h}_{3}$ given in~(\ref{s1}), obtained through a St\"ackel transform in Example~\ref{4}, generated by the separation relations (\ref{sep1}) with $k=-1$, that is \ by
\begin{gather*}
\tilde{h}_{1}\lambda_{i}^{-1}+\tilde{\alpha}\lambda_{i}^{2}+\tilde{h}_{2}\lambda_{i}+\tilde{h}_{3}=\frac{1}{2}\lambda_{i}\mu_{i}^{2}, \qquad
 i=1,2,3.
\end{gather*}
The metric associated with $\tilde{h}_{1}$%
\begin{gather}
\tilde{G}=\frac{1}{4}r_{3}^{2}\left(
\begin{matrix}
0 & 1 & 0\\
1 & 0 & 0\\
0 & 0 & 1
\end{matrix}
\right) \label{metric}
\end{gather}
is of constant curvature as~-- by Lemma~\ref{lemat}~-- after applying transformation (\ref{tran}), in the new separation coordinates the separation relations (\ref{sep1}) turns to
\begin{gather*}
\tilde{\alpha}\lambda_{i}^{-1}+\tilde{h}_{1}\lambda_{i}^{2}+\tilde{h} _{3}\lambda_{i}+\tilde{h}_{2}=\frac{1}{2}\lambda_{i}^{4}\mu_{i}^{2}, \qquad i=1,2,3
\end{gather*}
and belong again to the class~(\ref{BenSC}). Thus, by Remark~\ref{jakkwantowac}, we have to perform the minimal quantization of this system with respect to the original metric $\widetilde{A}_{1}$ of the system which is just~(\ref{metric}). Observing that $\sqrt{\vert \tilde{g} \vert}=8/r_{3}^{3}$, we obtain the following quantum superintegrable system (we use the second expression in~(\ref{hq})):
\begin{gather*}
\widehat{\tilde{h}}_{1} =-\frac{1}{4}\hslash^{2}r_{3}^{2}\left( \frac{1}{2}r_{3}\partial_{3}\frac{1}{r_{3}}\partial_{3}+\partial_{1}\partial
_{2}\right) -\frac{1}{4}\tilde{\alpha}r_{3}^{2},\\
\widehat{\tilde{h}}_{2} =\frac{1}{4}\hslash^{2}\left( -2\partial_{1}^{2}+2\partial_{2}r_{2}\partial_{2}+\partial_{1}r_{1}\partial_{2}
+r_{1}\partial_{2}\partial_{1}+r_{3}\partial_{2}\partial_{3}+r_{3}^{3}\partial_{3}\frac{1}{r_{3}^{2}}\partial_{2}\right) +\tilde{\alpha}r_{1},\\
\widehat{\tilde{h}}_{3} =\frac{1}{8}\hslash^{2}\left[ -r_{3}^{2}\partial_{2}^{2}+(r_{1}^{2}+4r_{2})\partial_{1}\partial_{2}+\partial_{1}
r_{1}^{2}\partial_{2}+4\partial_{2}r_{2}\partial_{1}+2r_{3}\partial_{1}\partial_{3}+2r_{3}^{3}\partial_{3}\frac{1}{r_{3}^{2}}\partial_{1}\right.\\
\left.\hphantom{\widehat{\tilde{h}}_{3} =}{} + r_{1}r_{3}\partial_{2}\partial_{3}+r_{1}r_{3}^{3}\partial_{3}\frac{1}{r_{3}^{2}}\partial_{1}\right] +\frac{1}{4}\tilde{\alpha}\big(r_{1}^{2}+4r_{2}\big) ,\\
\widehat{\tilde{h}}_{4} =\frac{1}{2}\hslash^{2}\partial_{2}^{2},\qquad \widehat{\tilde{h}}_{5} =\hslash^{2}\partial_{1}\partial_{2}+\tilde{\alpha}.
\end{gather*}
\end{Example}

\begin{Example}Let us f\/inally minimally quantize the St\"ackel Hamiltonians $\tilde{h}_{1}$, $\tilde{h}_{2}$, $\tilde{h}_{3}$ given in~(\ref{s3}), obtained through a~St\"ackel transform in Example~\ref{4}) and generated by separation relations~(\ref{sep1}) with $k=4$
\begin{gather*}
\tilde{h}_{1}\lambda_{i}^{4}+\tilde{\alpha}\lambda_{i}^{2}+\tilde{h}_{2}\lambda_{i}+\tilde{h}_{3}=\frac{1}{2}\lambda_{i}\mu_{i}^{2},\qquad i=1,2,3.
\end{gather*}
The metric associated with $\tilde{h}_{1}$
\begin{gather*}
\tilde{G}=\frac{1}{\frac{3}{4}r_{1}^{2}-r_{2}}\left(
\begin{matrix}
0 & 1 & 0\\
1 & 0 & 0\\
0 & 0 & 1
\end{matrix}
\right)
\end{gather*}
is conformally f\/lat. By Remark~\ref{jakkwantowac}, we have to perform minimal quantization of this system with respect to the metric~(\ref{MQ}) given by
\begin{gather*}
G=\big( -V_{1}^{(4)}\big) ^{1-\frac{2}{3}}\tilde{G}=\left( r_{2}-\frac{3}{4}r_{1}^{2}\right) ^{-\frac{2}{3}}\left(
\begin{matrix}
0 & 1 & 0\\
1 & 0 & 0\\
0 & 0 & 1
\end{matrix}
\right).
\end{gather*}
Observing that $\sqrt{\vert g\vert }=V_{1}^{(4)}=r_{2}-\frac{3}{4}r_{1}^{2}$, we obtain the following quantum operators (we use again the second expression in~(\ref{hq})):
\begin{gather*}
\widehat{\tilde{h}}_{1} =\frac{1}{2}\hslash^{2}\left( r_{2}-\frac{3}{4}r_{1}^{2}\right) ^{-1}\big( 2\partial_{1}\partial_{2}+\partial_{3}^{2}\big) +\frac{\tilde{\alpha}}{r_{2}-\frac{3}{4}r_{1}^{2}},\\
\widehat{\tilde{h}}_{2} =-\frac{1}{2}\hslash^{2}\left( r_{2}-\frac{3}{4}r_{1}^{2}\right) ^{-1}\sum_{i,j}\partial_{i}B_{2}^{ij}\partial_{j}-\tilde{\alpha}\frac{r_{1}r_{2}+\frac{1}{4}r_{1}^{3}+\frac{1}{4}r_{3}^{2}}{r_{2}-\frac{3}{4}r_{1}^{2}},\\
\widehat{\tilde{h}}_{3} =-\frac{1}{2}\hslash^{2}\left( r_{2}-\frac{3}{4}r_{1}^{2}\right) ^{-1}\sum_{i,j}\partial_{i}B_{3}^{ij}\partial_{j}
+\frac{1}{4}\tilde{\alpha}\frac{r_{1}r_{3}^{2}}{r_{2}-\frac{3}{4}r_{1}^{2}},\\
\widehat{\tilde{h}}_{4} =-\frac{1}{2}\hslash^{2}\left( r_{2}-\frac{3}{4}r_{1}^{2}\right) ^{-1}\left[ \partial_{1}r_{1}\partial_{2}+r_{1}\partial_{2}\partial_{1}-\partial_{2}\left( r_{2}-\frac{3}{4}r_{1}^{2}\right) \partial_{2}+r_{1}\partial_{3}^{2}\right] -\tilde{\alpha}\frac{r_{1}}{r_{2}-\frac{3}{4}r_{1}^{2}},\\
\widehat{\tilde{h}}_{5} =-\frac{1}{2}\hslash^{2}\partial_{3}^{2},
\end{gather*}
where
\begin{gather*}
B_{2} =\left(
\begin{matrix}r_{2}-\frac{3}{4}r_{1}^{2} & \frac{3}{2}r_{1}r_{2}-\frac{1}{8}r_{1}^{3}+\frac{1}{4}r_{3}^{2} & 0\vspace{1mm}\\
\frac{3}{2}r_{1}r_{2}-\frac{1}{8}r_{1}^{3}+\frac{1}{4}r_{3}^{2} &
-r_{2}\left( r_{2}-\frac{3}{4}r_{1}^{2}\right) & -\frac{1}{2}r_{3}\left(r_{2}-\frac{3}{4}r_{1}^{2}\right) \vspace{1mm}\\
0 & -\frac{1}{2}r_{3}\left( r_{2}-\frac{3}{4}r_{1}^{2}\right) & 2r_{1}r_{2}-\frac{1}{2}r_{1}^{3}+\frac{1}{4}r_{3}^{2}
\end{matrix}
\right), \\
B_{3} =\left(
\begin{matrix}
0 & -\frac{1}{4}r_{1}r_{3}^{2} & -\frac{1}{2}r_{3}\left( r_{2}-\frac{3}{4}r_{1}^{2}\right) \vspace{1mm}\\
-\frac{1}{4}r_{1}r_{3}^{2} & \frac{1}{4}r_{3}^{2}\left( r_{2}-\frac{3}{4}r_{1}^{2}\right) & -\frac{1}{4}r_{1}r_{3}\left( r_{2}-\frac{3}{4}r_{1}^{2}\right) \vspace{1mm}\\
-\frac{1}{2}r_{3}\left( r_{2}-\frac{3}{4}r_{1}^{2}\right) & -\frac{1}{4}r_{1}r_{3}\left( r_{2}-\frac{3}{4}r_{1}^{2}\right) & -\frac{1}{2}r_{1}^{2}r_{2}+r_{2}^{2}-\frac{1}{4}r_{1}r_{3}^{2}-\frac{3}{16}r_{1}^{4}
\end{matrix}
\right)
\end{gather*}
with $B=\sqrt{\vert g\vert }A$ in~(\ref{hq}). It can be checked that it is again a quantum superintegrable system.
\end{Example}

\pdfbookmark[1]{References}{ref}
\LastPageEnding

\end{document}